\documentclass[11pt,a4paper,intlimits]{amsart}

\usepackage{chicago}
\usepackage{amsmath}
\usepackage{amssymb}
\usepackage{latexsym}
\usepackage{enumerate}
\usepackage{xspace}

\usepackage[bookmarksopen,pdfstartview=FitH]{hyperref}
\hypersetup{colorlinks,%
            citecolor=black,%
            filecolor=black,%
            linkcolor=black,%
            urlcolor=black}

\newif\ifpdf
\ifx\pdfoutput\undefined
   \pdffalse        
\else
   \pdfoutput=1     
   \pdftrue
\fi

\ifpdf
   \usepackage[pdftex]{epsfig,rotating}
\else
   \usepackage{graphicx}
\fi

\begin{document}

\theoremstyle{plain}
\newtheorem{thm}{Theorem}[section]
\newtheorem{lem}[thm]{Lemma}
\newtheorem{cor}[thm]{Corollary}
\newtheorem{prop}[thm]{Proposition}

\theoremstyle{definition}
\newtheorem{defin}[thm]{Definition}
\newtheorem{example}[thm]{Example}
\newtheorem{rem}[thm]{Remark}
\newtheorem{character}[thm]{Characterization}
\newtheorem*{assum}{Assumption ($\mathbb{M}$)}
\newtheorem*{assuem}{Assumption ($\mathbb{EM}$)}
\def\E{\mathrm{I\kern-.2em E}}
\def\Et{\bar{\E}}
\def\P{\bar{P}}
\def\R{\mathbb{R}}
\def\C{\mathbb{C}}
\def\F{\mathcal{F}}
\def\bF{\mathbf{F}}

\def\ud{\mathrm{d}}
\def\dx{\mathrm{d}x}
\def\dt{\mathrm{d}t}
\def\ds{\mathrm{d}s}
\def\dsdx{\mathrm{d}s,\mathrm{d}x}

\def\1{1}
\def\e{\mathrm{e}}
\def\ap{\infty}

\def\eqlaw{\stackrel{\mathrm{d}}{=}}
\def\Lt{(L_t)_{0\leq t\leq T}}
\def\ott{{0\leq t\leq T}}

\def\lev{L\'evy\xspace}
\def\proc{process\xspace}
\def\procs{processes\xspace}
\def\cad{c\`{a}dl\`{a}g\xspace}
\newcommand{\prozess}[1][L]{{\ensuremath{#1=(#1_t)_{0\le t\le T}}}\xspace}
\newcommand{\prazess}[1][L]{{\ensuremath{#1=(#1_t)_{0\le t\le T^*}}}\xspace}

\newcommand{\Log}{\ensuremath{\mathop{\mathcal{L}\mathrm{og}}}}

\renewcommand{\thefigure}{\thesection.\arabic{figure}}

\renewcommand{\theequation}{\thesection.\arabic{equation}}
\numberwithin{equation}{section}

\newcommand\llambda{{\mathchoice
 {\lambda\mkern-4.5mu{\raisebox{.4ex}{\scriptsize$\backslash$}}}
 {\lambda\mkern-4.83mu{\raisebox{.4ex}{\scriptsize$\backslash$}}}
 {\lambda\mkern-4.5mu{\raisebox{.2ex}{\footnotesize$\scriptscriptstyle\backslash$}}}
 {\lambda\mkern-5.0mu{\raisebox{.2ex}{\tiny$\scriptscriptstyle\backslash$}}}}}

\frenchspacing

\title[introduction to L\'evy processes]{An introduction to L\'evy processes\\with applications in Finance}

\author{Antonis Papapantoleon}

\address{Financial and Actuarial Mathematics, Vienna University of Technology,
         Wiedner Hauptstrasse 8/105, 1040 Vienna, Austria }
\email{papapan@fam.tuwien.ac.at}
\urladdr{http://www.fam.tuwien.ac.at/\~{}papapan}



\thanks{These lecture notes were prepared for mini-courses taught at the University
        of Piraeus in April 2005 and March 2008, at the University of Leipzig in
        November 2005 and at the Technical University of Athens in September 2006 and
        March 2008. I am grateful for the opportunity of lecturing on these topics to
        George Skiadopoulos, Thorsten Schmidt, Nikolaos Stavrakakis and Gerassimos
        Athanassoulis.}

\maketitle

\pagestyle{myheadings}

\begin{abstract}
These lectures notes aim at introducing \lev processes in an
informal and intuitive way, accessible to non-specialists in the
field. In the first part, we focus on the theory of \lev processes.
We analyze a `toy' example of a \lev process, viz. a \lev
jump-diffusion, which yet offers significant insight into the
distributional and path structure of a \lev process. Then, we
present several important results about \lev processes, such as
infinite divisibility and the L\'evy-Khintchine formula, the
L\'evy-It\^o decomposition, the It\^o formula for \lev \procs and
Girsanov's transformation. Some (sketches of) proofs are presented,
still the majority of proofs is omitted and the reader is referred
to textbooks instead. In the second part, we turn our attention to
the applications of \lev processes in financial modeling and option
pricing. We discuss how the price process of an asset can be modeled
using \lev \procs and give a brief account of market incompleteness.
Popular models in the literature are presented and revisited from
the point of view of \lev \procs, and we also discuss three methods
for pricing financial derivatives. Finally, some indicative evidence
from applications to market data is presented.
\end{abstract}

\tableofcontents

\part{Theory}

\section{Introduction}

\lev processes play a central role in several fields of science,
such as \textit{physics}, in the study of turbulence, laser cooling
and in quantum field theory; in \textit{engineering}, for the study
of networks, queues and dams; in \textit{economics}, for continuous
time-series models; in the \textit{actuarial science}, for the
calculation of insurance and re-insurance risk; and, of course, in
\textit{mathematical finance}. A comprehensive overview of several
applications of \lev processes can be found in \citeN{Prabhu98}, in
\citeN{Barndorff-NielsenMikoschResnick01}, in
\citeN{KyprianouSchoutensWilmott05} and in \citeN{Kyprianou06}.

In mathematical finance, \lev processes are becoming extremely
fashionable because they can describe the observed reality of
financial markets in a more accurate way than models based on
Brownian motion. In the `real' world, we observe that asset price
processes have jumps or spikes, and risk managers have to take them
into consideration; in Figure \ref{usdjpy} we can observe some big
price changes (jumps) even on the very liquid USD/JPY exchange rate.
Moreover, the empirical distribution of asset returns exhibits fat
tails and skewness, behavior that deviates from normality; see
Figure \ref{returns} for a characteristic picture. Hence, models
that accurately fit return distributions are essential for the
estimation of profit and loss (P\&L) distributions. Similarly, in
the `risk-neutral' world, we observe that implied volatilities are
constant neither across strike nor across maturities as stipulated
by the \citeN{BlackScholes73} (actually, \citeNP{Samuelson65})
model; Figure \ref{surface} depicts a typical volatility surface.
Therefore, traders need models that can capture the behavior of the
implied volatility smiles more accurately, in order to handle the
risk of trades. \lev processes provide us with the appropriate tools
to adequately and consistently describe all these observations, both
in the `real' and in the `risk-neutral' world.

\begin{figure}
\begin{center}
 \includegraphics[width=8cm,keepaspectratio=true]{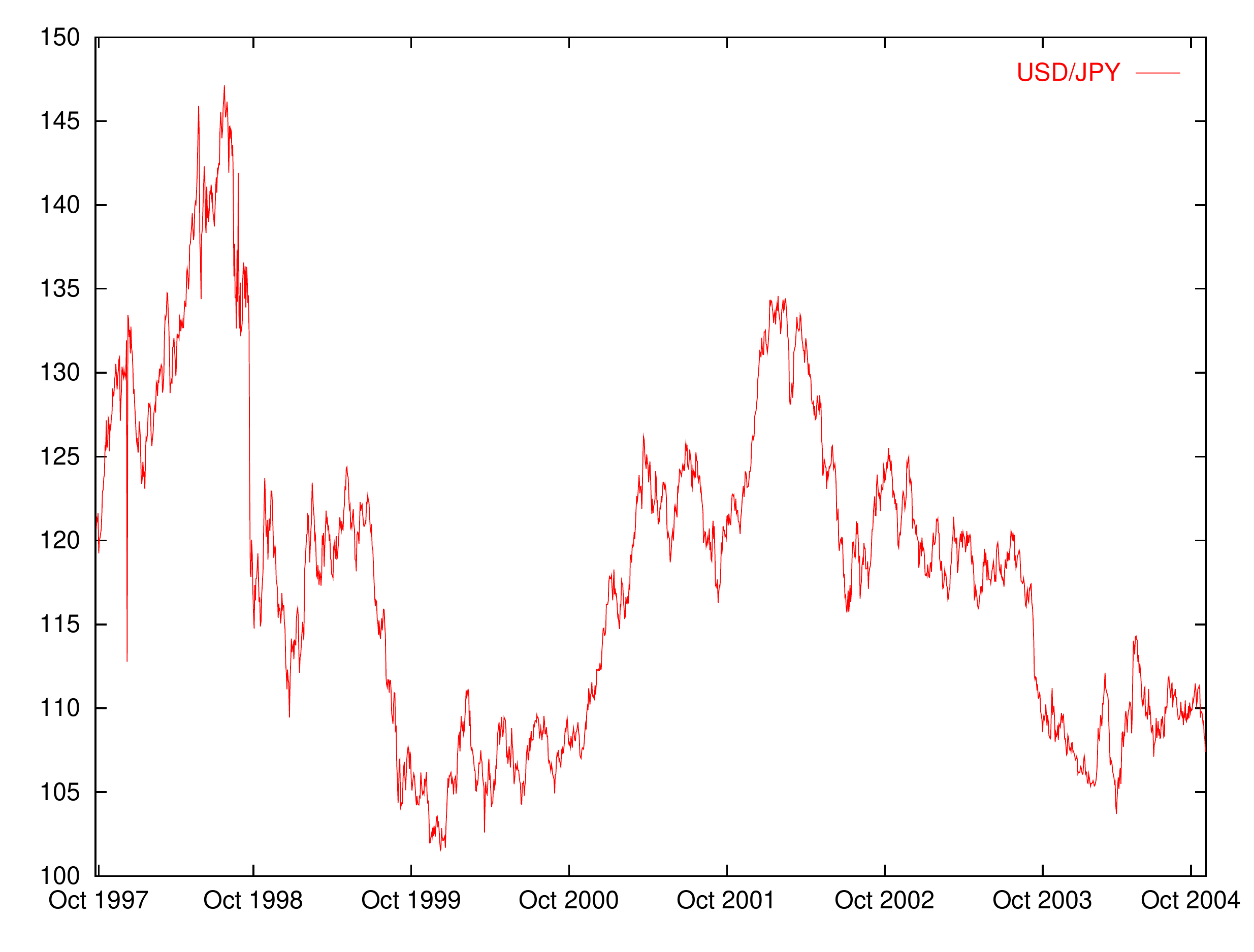}
  \caption{USD/JPY exchange rate, Oct. 1997--Oct. 2004.}
   \label{usdjpy}
\end{center}
\end{figure}

The main aim of these lecture notes is to provide an accessible
overview of the field of \lev \procs and their applications in
mathematical finance to the non-specialist reader. To serve that
purpose, we have avoided most of the proofs and only sketch a number
of proofs, especially when they offer some important insight to the
reader. Moreover, we have put emphasis on the intuitive
understanding of the material, through several pictures and
simulations.

We begin with the definition of a \lev \proc and some known
examples. Using these as the reference point, we construct and study
a \lev jump-diffusion; despite its simple nature, it offers
significant insights and an intuitive understanding of general \lev
\procs. We then discuss infinitely divisible distributions and
present the celebrated \lev--Khintchine formula, which links
processes to distributions. The opposite way, from distributions to
processes, is the subject of the \lev-It\^o decomposition of a \lev
\proc. The \lev measure, which is responsible for the richness of
the class of \lev processes, is studied in some detail and we use it
to draw some conclusions about the path and moment properties of a
\lev \proc. In the next section, we look into several subclasses
that have attracted special attention and then present some
important results from semimartingale theory. A study of martingale
properties of \lev \procs and the It\^o formula for \lev \procs
follows. The change of probability measure and Girsanov's theorem
are studied is some detail and we also give a complete proof in the
case of the Esscher transform. Next, we outline three ways for
constructing new \lev \procs and the first part closes with an
account on simulation methods for some \lev \procs.

\begin{figure}
 \begin{center}
  \includegraphics[width=8cm,keepaspectratio=true]{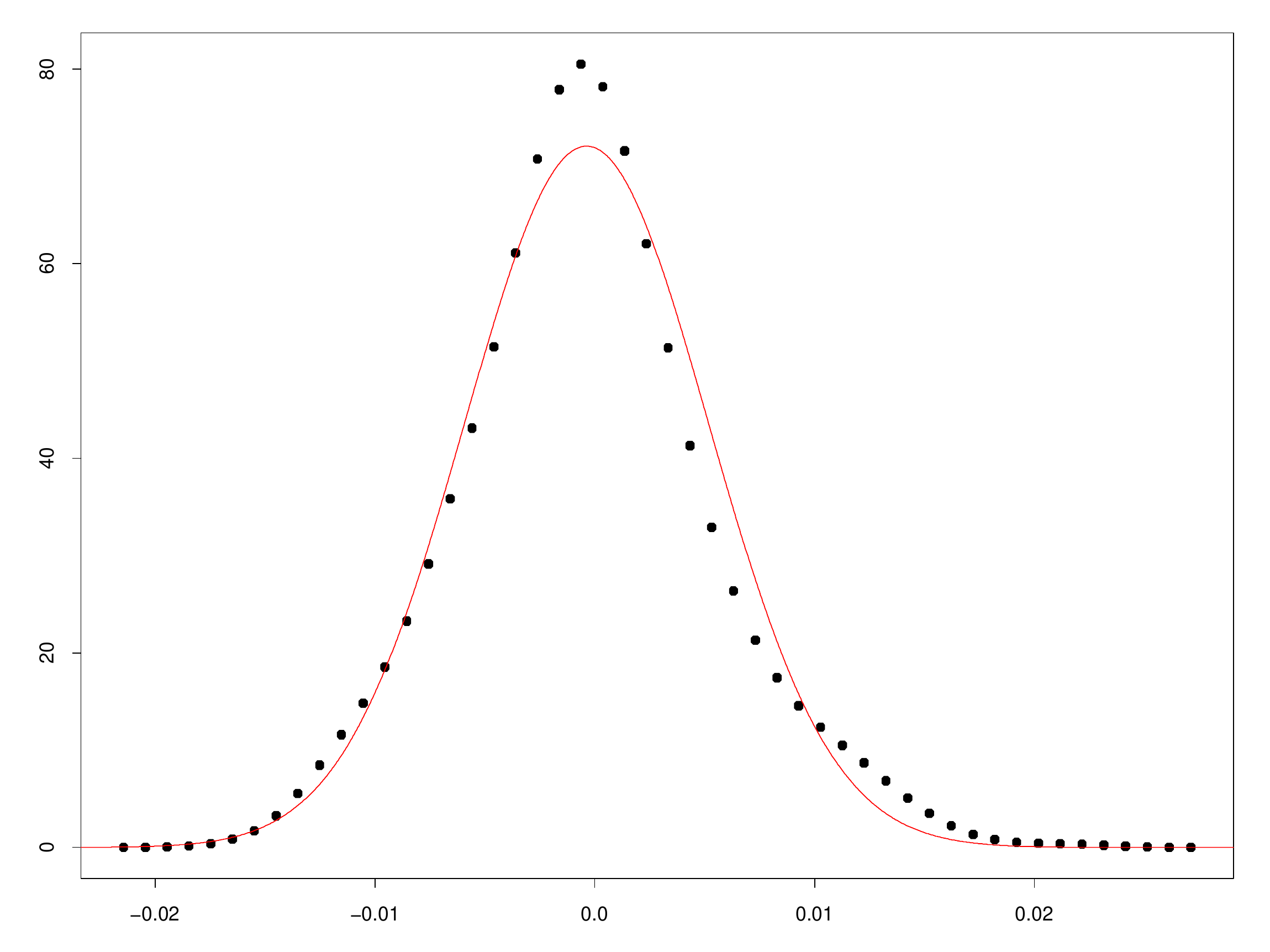}
  \caption{Empirical distribution of daily log-returns for the
            GBP/USD exchange rate and fitted Normal distribution.}
    \label{returns}
 \end{center}
\end{figure}

The second part of the notes is devoted to the applications of \lev
\procs in mathematical finance. We describe the possible approaches
in modeling the price process of a financial asset using \lev \procs
under the `real' and the `risk-neutral' world, and give a brief
account of market incompleteness which links the two worlds. Then,
we present a primer of popular \lev models in the mathematical
finance literature, listing some of their key properties, such as
the characteristic function, moments and densities (if known). In
the next section, we give an overview of three methods for pricing
options in \lev-driven models, viz. transform, partial
integro-differential equation (PIDE) and Monte Carlo methods.
Finally, we present some empirical results from the application of
\lev \procs to real market financial data. The appendices collect
some results about Poisson random variables and processes, explain
some notation and provide information and links regarding the data
sets used.

Naturally, there is a number of sources that the interested reader
should consult in order to deepen his knowledge and understanding of
\lev \procs. We mention here the books of \citeN{Bertoin96},
\citeN{Sato99}, Applebaum \citeyear{Applebaum04}, \citeN{Kyprianou06} on various
aspects of \lev \procs. Cont and Tankov \citeyear{ContTankov03} and \citeN{Schoutens03}
focus on the applications of L\'evy processes in finance. The books
of \citeN{JacodShiryaev03} and Protter \citeyear{Protter04} are essential
readings for semimartingale theory, while \citeN{Shiryaev99} blends
semimartingale theory and applications to finance in an impressive
manner. Other interesting and inspiring sources are the papers by
\citeN{Eberlein01a}, \citeN{Cont01},
\citeN{Barndorff-NielsenPrause01}, Carr et al. \citeyear{Carretal02},
Eberlein and \"Ozkan\citeyear{EberleinOezkan03} and \citeN{Eberlein07}.

\begin{figure}
 \begin{center}
  \includegraphics[width=8.75cm,keepaspectratio=true]{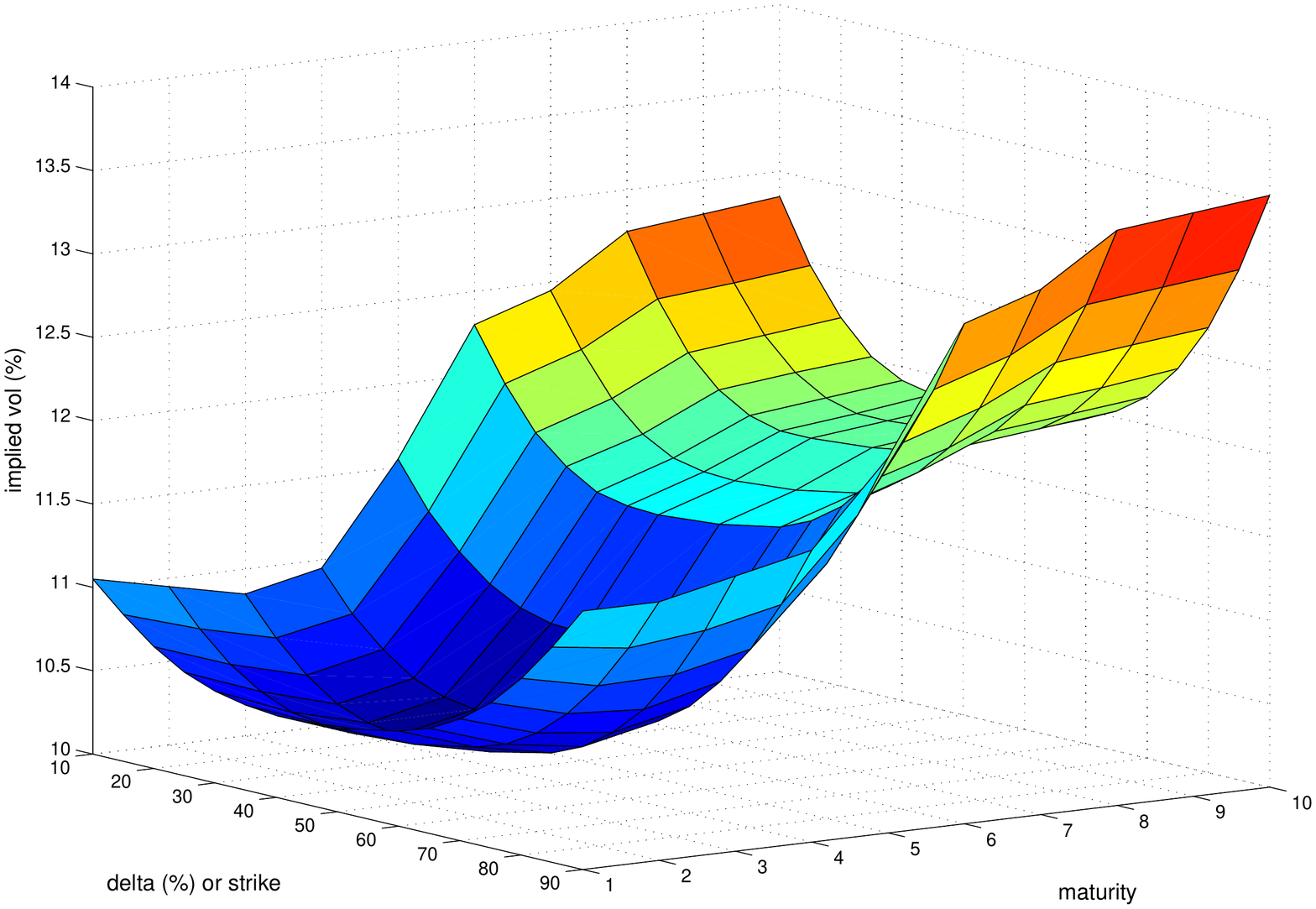}
   \caption{Implied volatilities of vanilla options on the
            EUR/USD exchange rate on November 5, 2001.}
    \label{surface}
 \end{center}
\end{figure}

\section{Definition}

Let ($\Omega, \mathcal{F}, \bF, P$) be a filtered probability space,
where $\F=\F_T$ and the filtration $\bF=(\F_t)_{t\in[0,T]}$
satisfies the usual conditions. 
Let $T\in[0,\ap]$ denote the time horizon which, in general, can be
infinite.

\begin{defin}
A c\`{a}dl\`{a}g, adapted, real valued stochastic process \prozess
with $L_{0}=0$ a.s. is called a \emph{L\'evy process} if the
following conditions are satisfied:
\begin{description}
    \item[(L1)] $L$ has \emph{independent increments}, i.e.
                $L_t - L_{s}$ is independent of $\F_s$ for any $0\leq s<t\leq T$.
    \item[(L2)] $L$ has \emph{stationary increments}, i.e.
                for any $0\leq s,t\leq T$ the distribution of $L_{t+s} - L_{t}$ does
                not depend on $t$.
    \item[(L3)] $L$ is \emph{stochastically continuous}, i.e.
                for every $\ott$ and $\epsilon>0$:
                $\lim_{s\rightarrow t}P(|L_t-L_s|>\epsilon)=0$.
\end{description}
\end{defin}

The simplest \lev process is the linear drift, a deterministic
process. Brownian motion is the only (non-deterministic) \lev
process with continuous sample paths.
\begin{figure}
\begin{center}
 \includegraphics[height=6cm,keepaspectratio=true,angle=90]{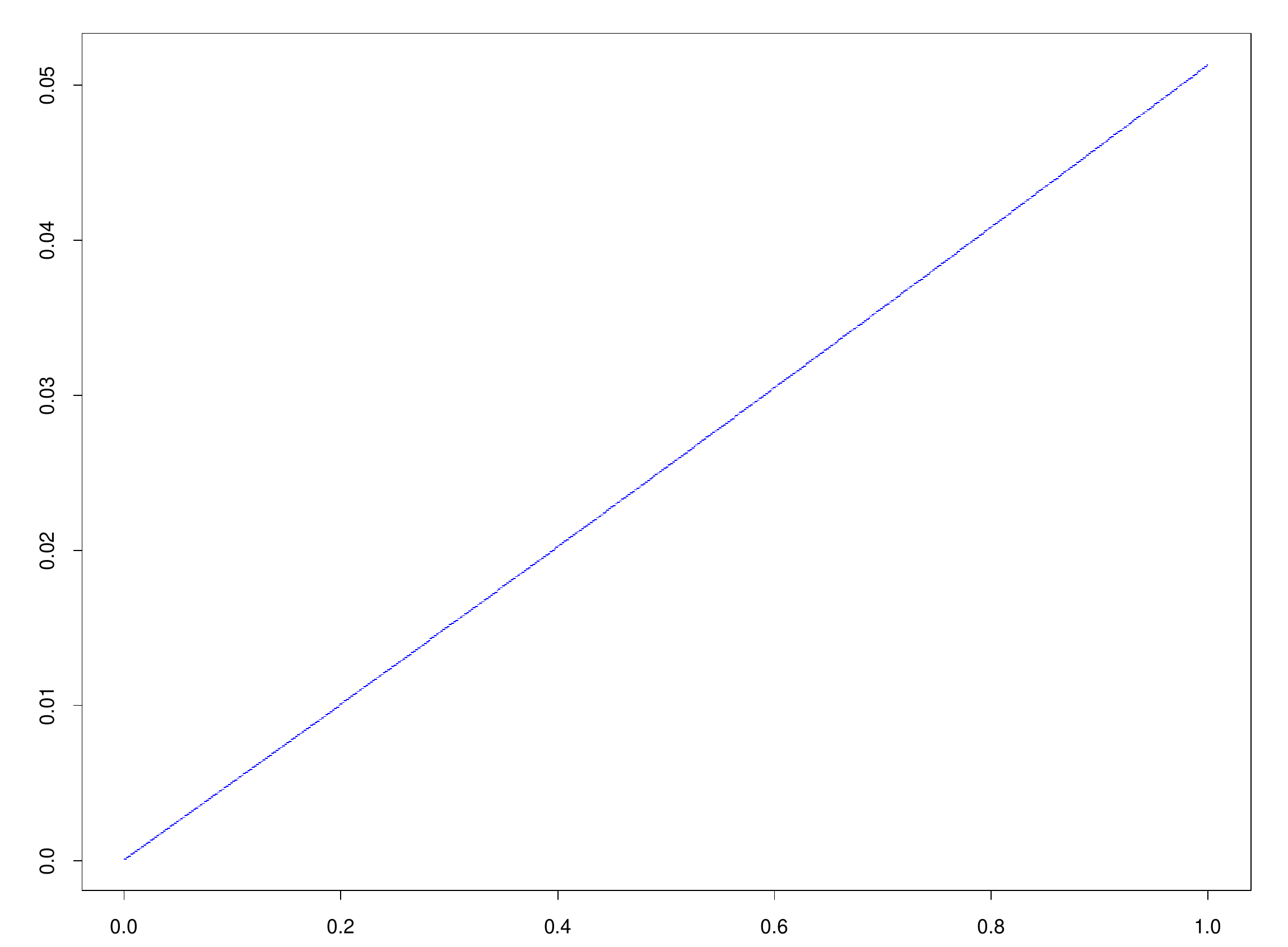}
 \includegraphics[width=6cm,keepaspectratio=true]{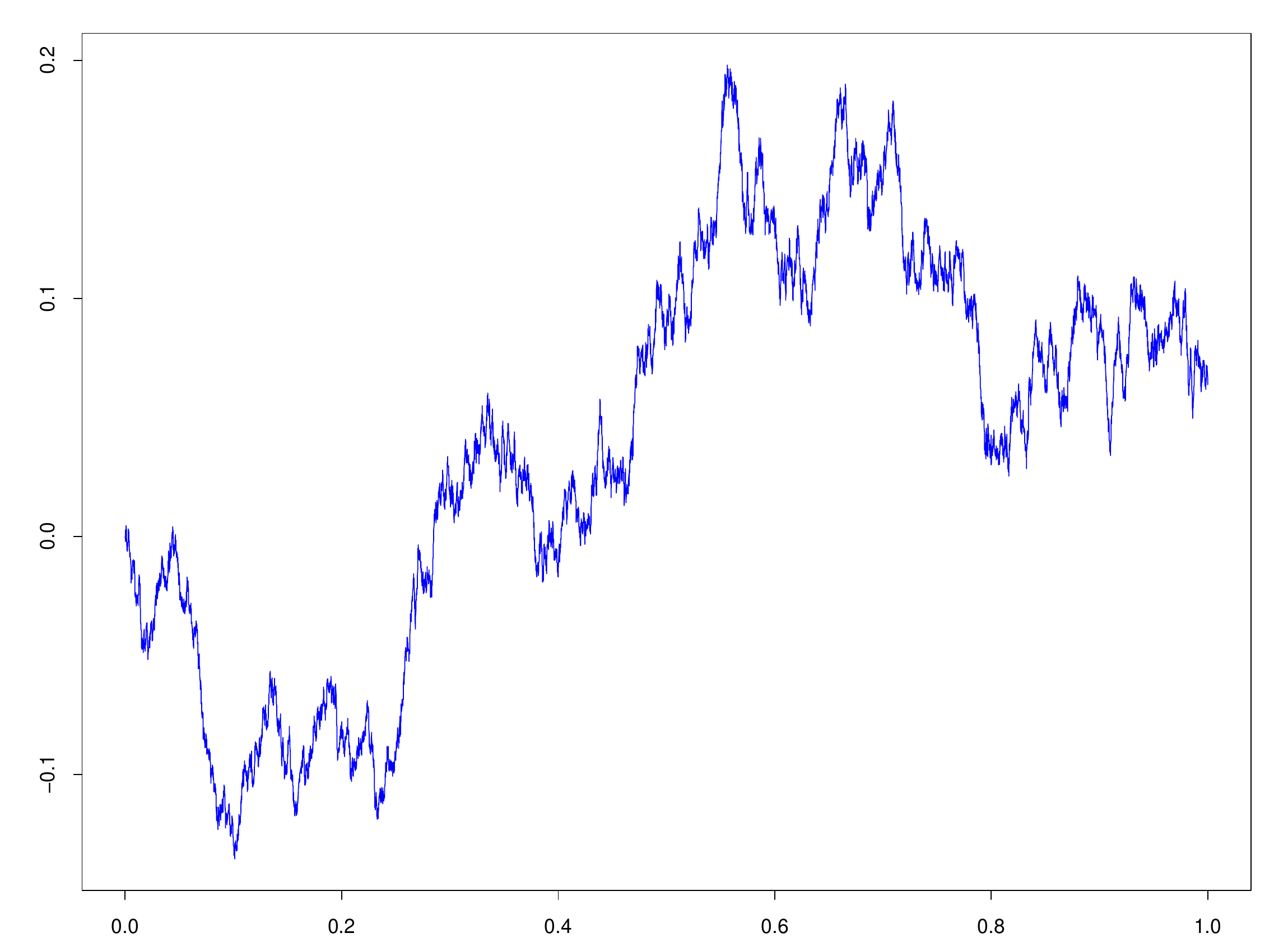}
  \caption{Examples of \lev processes: linear drift (left) and Brownian motion.}
\end{center}
\end{figure}
Other examples of \lev processes are the Poisson and compound
Poisson processes. Notice that the sum of a linear drift, a Brownian
motion and a compound Poisson process is again a \lev process; it is
often called a ``jump-diffusion'' process. We shall call it a
``\emph{\lev jump-diffusion}'' process, since there exist
jump-diffusion processes which are \textit{not} \lev processes.
\begin{figure}
\begin{center}
 \includegraphics[height=6cm,keepaspectratio=true,angle=90]{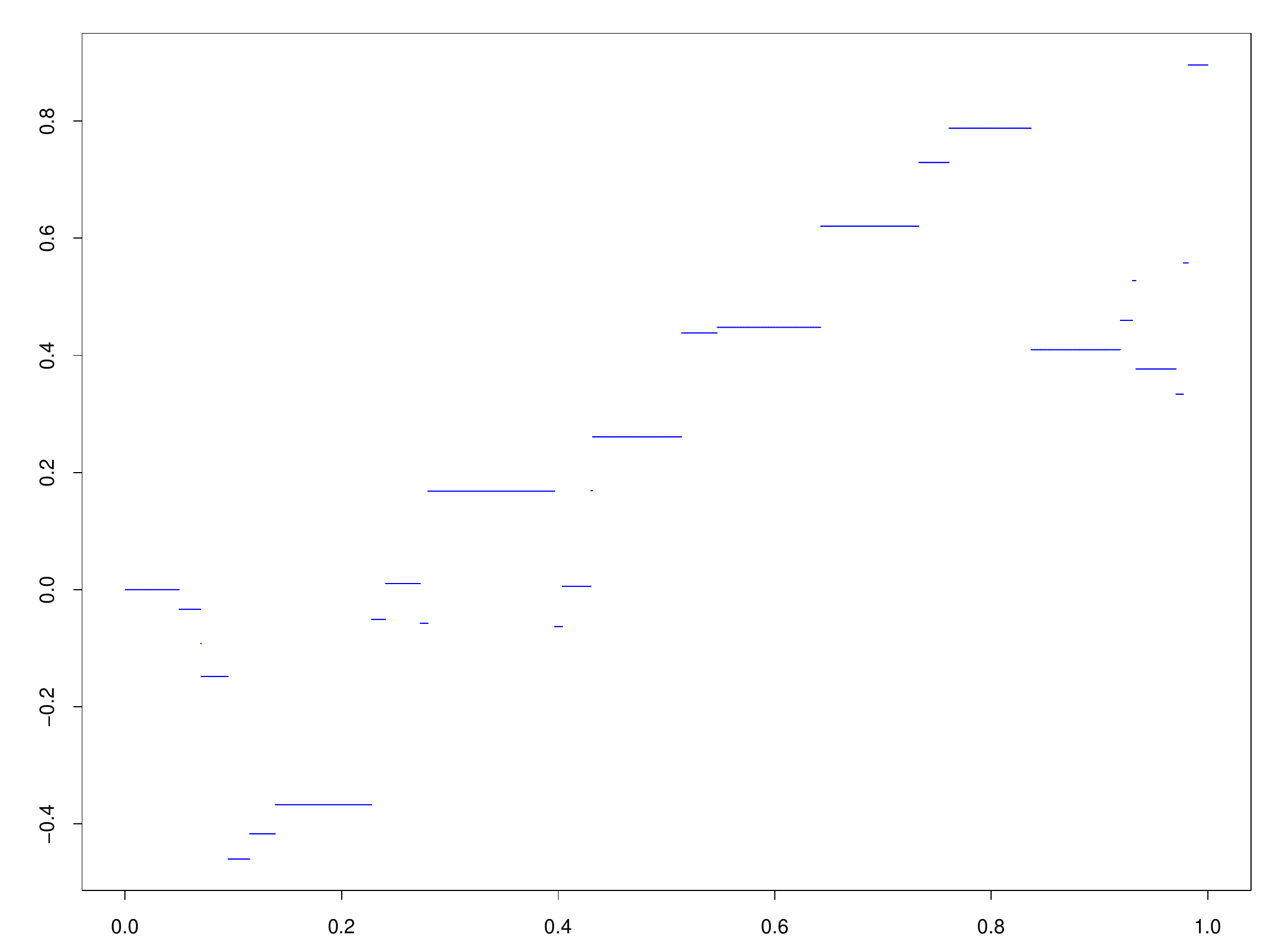}
 \includegraphics[width=6cm,keepaspectratio=true]{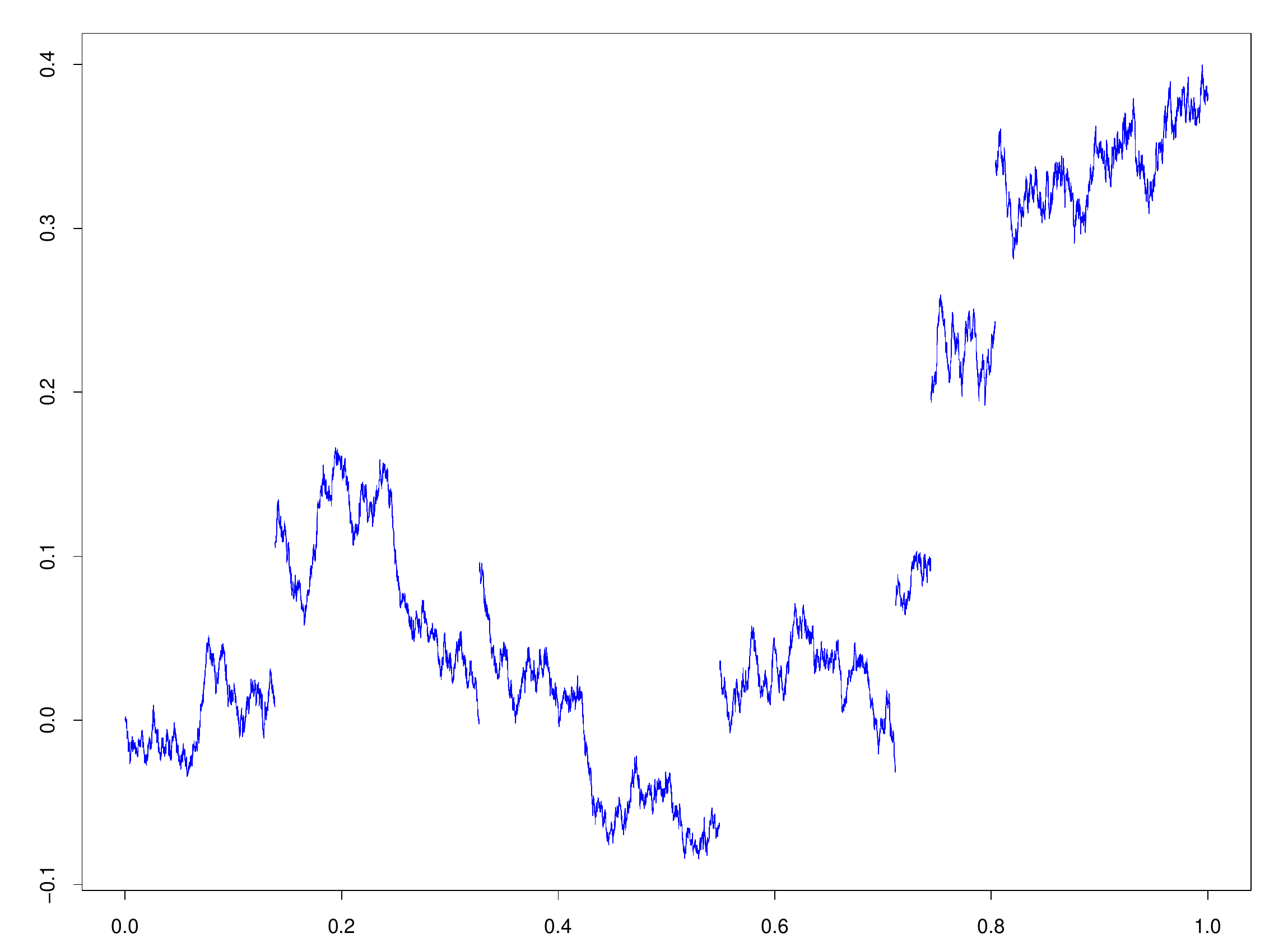}
  \caption{Examples of \lev processes: compound Poisson process (left) and \lev jump-diffusion.}
\end{center}
\end{figure}

\section{`Toy' example: a L\'evy jump-diffusion}

Assume that the process $L=\Lt$ is a \lev jump-diffusion, i.e. a
Brownian motion plus a compensated compound Poisson process. The
paths of this process can be described by
\begin{align}\label{EquationBasis}
L_{t}= b t + \sigma W_{t}+ \Big(\sum_{k=1}^{N_{t}} J_{k} - t\lambda\kappa\Big)
\end{align}
where $b\in\R$, $\sigma\in\R_{\geqslant0}$, $W=(W_t)_\ott$ is a
standard Brownian motion, $N=(N_t)_\ott$ is a Poisson process with
parameter $\lambda$ (i.e. $\E[N_t]=\lambda t$) and
$J=(J_k)_{k\geq1}$ is an i.i.d. sequence of random variables with
probability distribution $F$ and $\E[J]=\kappa<\infty$. Hence, $F$
describes the distribution of the jumps, which arrive according to
the Poisson process. All sources of randomness are \textit{mutually
independent}.

It is well known that Brownian motion is a martingale; moreover, the
compensated compound Poisson process is a martingale. Therefore,
$L=\Lt$ is a \textit{martingale} if and only if $b=0$.

The characteristic function of $L_t$ is
\begin{align*}
\E\big[\e^{iu L_{t}}\big]
 &= \E\Big[\exp\big(iu \big(b t + \sigma W_t
   + \sum_{k=1}^{N_{t}} J_{k}-t\lambda\kappa\big)\big)\Big]\\
 &= \exp\big[iub t\big]
    \E\Big[\exp\big(iu\sigma W_t\big)
       \exp\big(iu\big(\sum_{k=1}^{N_{t}} J_{k} - t\lambda\kappa\big)\big)\Big];\\
 \intertext{since all the sources of randomness are independent, we get}
 &= \exp\big[iu b t\big]
    \E\Big[\exp\big(iu \sigma W_t\big)\Big]
    \E\Big[\exp\big(iu\sum_{k=1}^{N_{t}} J_{k}-iut\lambda\kappa\big)\Big];\\
 \intertext{taking into account that
\begin{align*}
\E[\e^{iu\sigma W_t}]=\e^{-\frac{1}{2}\sigma^{2}u^{2}t}, \quad{}
        W_t \sim \text{Normal}(0,t)
\end{align*}
\begin{align*}
\E[\e^{iu\sum_{k=1}^{N_t}J_k}]=\e^{\lambda t(\E[\e^{iuJ}-1])}, \quad{}
        N_t \sim \text{Poisson} (\lambda t)
\end{align*}
    (cf. also Appendix \ref{comp-Poisson}) we get}
 &= \exp\big[ iub t\big]
    \exp\Big[-\frac{1}{2}u^{2}\sigma^{2}t\Big]
    \exp\Big[\lambda t\big(\E[\e^{iu J}-1]-iu\E[J]\big)\Big]\\
 &= \exp\big[ iub t\big]
    \exp\Big[-\frac{1}{2}u^{2}\sigma^{2}t\Big]
    \exp\Big[\lambda t\big(\E[\e^{iu J}-1-iuJ]\big)\Big];\\
\intertext{and because the distribution of $J$ is $F$ we have}
 &= \exp\big[iub t\big]
    \exp\Big[-\frac{1}{2}u^2\sigma^{2}t\Big]
    \exp\Big[\lambda t \int_{\R}\big(\e^{iux}-1-iux\big)F(\dx)\Big].
\end{align*}

Now, since $t$ is a common factor, we re-write the above equation as
\begin{align}\label{lk-ljd}
\E\big[\e^{iu L_{t}}\big] =
 \exp\bigg[t \Big(
      iub - \frac{u^2\sigma^{2}}{2}+
      \int_{\R} (\e^{iux}-1-iux)\lambda F(\ud x)
     \Big)\bigg].
\end{align}

Since the characteristic function of a random variable determines
its distribution, we have a ``characterization" of the distribution
of the random variables underlying the \lev jump-diffusion. We will
soon see that this distribution belongs to the class of
\emph{infinitely divisible distributions} and that equation
(\ref{lk-ljd}) is a special case of the celebrated
\emph{L\'evy-Khintchine formula}.

\begin{rem}
Note that time factorizes out, and the drift, diffusion and jumps
parts are separated; moreover, the jump part factorizes to expected
number of jumps ($\lambda$) and distribution of jump size ($F$). It
is only natural to ask if these features are preserved for
\textit{all} \lev processes. The answer is yes for the first two
questions, but jumps cannot be always separated into a product of
the form $\lambda\times F$.
\end{rem}

\section{Infinitely divisible distributions and the L\'evy-Khintchine formula}

There is a strong interplay between \lev processes and infinitely
divisible distributions. We first define infinitely divisible
distributions and give some examples, and then describe their
relationship to \lev processes.

Let $X$ be a real valued random variable, denote its characteristic
function by $\varphi_X$ and its law by $P_X$, hence
$\varphi_X(u)=\int_{\R}\e^{iux}P_X(\ud x)$. Let $\mu*\nu$ denote the
convolution of the measures $\mu$ and $\nu$, i.e.
$(\mu*\nu)(A)=\int_\R\nu(A-x)\mu(\dx)$.

\begin{defin}
The law $P_X$ of a random variable $X$ is \emph{infinitely
divisible}, if for all $n\in\mathbb{N}$ there exist i.i.d. random
variables $X^{(1/n)}_1,\ldots,X^{(1/n)}_n$ such that
\begin{align}
X \eqlaw X^{(1/n)}_1+ \ldots +X^{(1/n)}_n.
\end{align}
Equivalently, the law $P_X$ of a random variable $X$ is
\emph{infinitely divisible} if for all $n\in\mathbb{N}$ there exists
another law $P_{X^{(1/n)}}$ of a random variable $X^{(1/n)}$ such
that
\begin{align}
 P_X = \underbrace{P_{X^{(1/n)}}*\ldots*P_{X^{(1/n)}}}_{n\; \text{times}}.
\end{align}
\end{defin}

Alternatively, we can characterize an infinitely divisible random
variable using its characteristic function.

\begin{character}
The law of a random variable $X$ is \emph{infinitely divisible}, if
for all $n\in\mathbb{N}$, there exists a random variable
$X^{(1/n)}$, such that
\begin{align}
\varphi_X(u)=\Big(\varphi_{X^{(1/n)}}(u)\Big)^n.
\end{align}
\end{character}

\begin{example}[Normal distribution]
Using the characterization above, we can easily deduce that the
Normal distribution is infinitely divisible. Let
$X\sim\text{Normal}(\mu,\sigma^2)$, then we have
\begin{align*}
\varphi_X(u) &= \exp \Big[iu\mu- \frac{1}{2}u^2\sigma^{2}\Big]
              = \exp \Big[n (iu\frac{\mu}{n}- \frac{1}{2}u^2\frac{\sigma^{2}}{n})\Big]\\
             &= \Bigg(\exp \Big[iu\frac{\mu}{n}- \frac{1}{2}u^2\frac{\sigma^{2}}{n}\Big]\Bigg)^n\\
             &= \Big(\varphi_{X^{(1/n)}}(u)\Big)^n,
\end{align*}
where
$X^{(1/n)}\sim\text{Normal}(\frac{\mu}{n},\frac{\sigma^2}{n})$.
\end{example}

\begin{example}[Poisson distribution]
Following the same procedure, we can easily conclude that the
Poisson distribution is also infinitely divisible. Let
$X\sim\text{Poisson}(\lambda)$, then we have
\begin{align*}
\varphi_X(u) &= \exp \Big[ \lambda(\e^{iu}-1) \Big]
              = \Bigg( \exp \Big[ \frac{\lambda}{n}(\e^{iu}-1) \Big]\Bigg)^n\\
             &= \Big(\varphi_{X^{(1/n)}}(u)\Big)^n,
\end{align*}
where $X^{(1/n)}\sim\text{Poisson}(\frac{\lambda}{n})$.
\end{example}

\begin{rem}
Other examples of infinitely divisible distributions are the
compound Poisson distribution, the exponential, the
$\Gamma$-distribution, the geometric, the negative binomial, the
Cauchy distribution and the strictly stable distribution.
Counter-examples are the uniform and binomial distributions.
\end{rem}

The next theorem provides a complete characterization of random
variables with infinitely divisible distributions via their
characteristic functions; this is the celebrated
\emph{L\'evy-Khintchine formula}. We will use the following
preparatory result (cf. \citeNP[Lemma 7.8]{Sato99}).

\begin{lem}\label{limid}
If $(P_k)_{k\ge0}$ is a sequence of infinitely divisible laws and
$P_k\rightarrow P$, then $P$ is also infinitely divisible.
\end{lem}

\begin{thm}
The law $P_X$ of a random variable $X$ is \emph{infinitely
divisible} if and only if there exists a triplet $(b,c,\nu)$, with
$b\in\R$, $c\in\R_{\geqslant0}$ and a measure satisfying
$\nu(\{0\})=0$ and $\int_{\R}(1\wedge|x|^2)\nu(\ud x)<\infty$, such
that
\begin{align}\label{lk-infdiv}
\E[\e^{iuX}]  
 = \exp\Big[ ibu - \frac{u^2c}{2}
   + \int_{\R}(\e^{iux}-1-iux\1_{\{|x|<1\}})\nu(\ud x) \Big].
\end{align}
\end{thm}
\begin{proof}[Sketch of Proof]
Here we describe the proof of the ``if" part, for the full proof see
Theorem 8.1 in \citeN{Sato99}. Let
$(\varepsilon_n)_{n\in\mathbb{N}}$ be a sequence in $\R$, monotonic
and decreasing to zero. Define for all $u\in\R$ and $n\in\mathbb{N}$
\begin{align*}
\varphi_{X_n}(u) =
 \exp\Big[iu\Big(b-\!\!\int_{\varepsilon_n<|x|\le1}\!\!x\nu(\ud x)\Big)
          -\frac{u^2c}{2}+ \int_{|x|>\varepsilon_n}(\e^{iux}-1)\nu(\dx)\Big].
\end{align*}
Each $\varphi_{X_n}$ is the convolution of a normal and a compound
Poisson distribution, hence $\varphi_{X_n}$ is the characteristic
function of an infinitely divisible probability measure $P_{X_n}$.
We clearly have that
\begin{align*}
 \lim_{n\rightarrow\infty}\varphi_{X_n}(u)=\varphi_{X}(u);
\end{align*}
then, by \lev's continuity theorem and Lemma \ref{limid},
$\varphi_{X}$ is the characteristic function of an infinitely
divisible law, provided that $\varphi_{X}$ is continuous at $0$.

Now, continuity of $\varphi_{X}$ at $0$ boils down to the continuity
of the integral term, i.e.
\begin{align*}
\psi_\nu(u)
 &= \int_{\R}(\e^{iux}-1-iux\1_{\{|x|<1\}})\nu(\ud x)\\
 &= \int_{\{|x|\le1\}}(\e^{iux}-1-iux)\nu(\ud x) + \int_{\{|x|>1\}}(\e^{iux}-1)\nu(\ud x).
\end{align*}
Using Taylor's expansion, the Cauchy--Schwarz inequality, the
definition of the \lev measure and dominated convergence, we get
\begin{align*}
|\psi_\nu(u)|
 &\le \frac12\int_{\{|x|\le1\}}|u^2x^2|\nu(\ud x) + \int_{\{|x|>1\}}|\e^{iux}-1|\nu(\ud x)\\
 &\le \frac{|u|^2}2\int_{\{|x|\le1\}}|x^2|\nu(\ud x) + \int_{\{|x|>1\}}|\e^{iux}-1|\nu(\ud x)\\
 & \longrightarrow0 \quad \text{as} \quad u\rightarrow0. \qedhere
\end{align*}
\end{proof}

The triplet ($b,c,\nu$) is called the \emph{L\'evy} or
\emph{characteristic triplet} and the exponent in \eqref{lk-infdiv}
\begin{align}\label{levy-exponent}
\psi(u)= iub - \frac{u^2c}{2} +
         \int_{\R}(\e^{iux}-1-iux\1_{\{|x|<1\}})\nu(\ud x)
\end{align}
is called the \emph{L\'evy} or \emph{characteristic exponent}.
Moreover, $b\in\R$ is called the \emph{drift term}, $c\in
\R_{\geqslant0}$ the \emph{Gaussian} or \emph{diffusion coefficient}
and $\nu$ the \emph{L\'evy measure}.

\begin{rem}
Comparing equations \eqref{lk-ljd} and \eqref{lk-infdiv}, we can
immediately deduce that the random variable $L_t$ of the \lev
jump-diffusion is infinitely divisible with \lev triplet $b=b\cdot
t,c=\sigma^2\cdot t$ and $\nu=(\lambda F)\cdot t$.
\end{rem}

Now, consider a \lev process $L=\Lt$; for any $n\in\mathbb N$ and
any $0<t\leq T$ we trivially have that
\begin{align}\label{simple}
L_t = L_{\frac{t}{n}}+(L_{\frac{2t}{n}}-L_{\frac{t}{n}})+ \ldots
+(L_{t}-L_{\frac{(n-1)t}{n}}).
\end{align}
The \emph{stationarity} and \emph{independence} of the increments
yield that $(L_{\frac{tk}{n}}-L_{\frac{t(k-1)}{n}})_{k\ge1}$ is an
i.i.d. sequence of random variables, hence we can conclude that the
\emph{random variable} $L_t$ is \emph{infinitely divisible}.

\begin{thm}
For every \lev process \prozess[L], we have that
\begin{align}
\E[\e^{iuL_t}]
 &= \e^{t\psi(u)}\\\nonumber
 &= \exp\Big[ t \big( ibu - \frac{u^2c}{2}
             + \int_{\R}(\e^{iux}-1-iux1_{\{|x|<1\}})\nu(\ud x)\big)\Big]
\end{align}
where $\psi(u)$ is the characteristic exponent of $L_1$, a random
variable with an infinitely divisible distribution.
\end{thm}
\begin{proof}[Sketch of Proof]
Define the function $\phi_u(t)=\varphi_{L_t}(u)$, then we have
\begin{align}\label{CFE}
\phi_u(t+s)
 &= \E[\e^{iuL_{t+s}}]
  = \E[\e^{iu(L_{t+s}-L_s)}\e^{iuL_s}]\\\nonumber
 &= \E[\e^{iu(L_{t+s}-L_s)}]\E[\e^{iuL_s}]
  = \phi_u(t)\phi_u(s).
\end{align}
Now, $\phi_u(0)=1$ and the map $t\mapsto\phi_u(t)$ is continuous (by
stochastic continuity). However, the unique continuous solution of
the Cauchy functional equation \eqref{CFE} is
\begin{align}
\phi_u(t)=\e^{t \vartheta(u)}, \qquad{ \mathrm{where} } \qquad
 \vartheta:\R\rightarrow\mathbb{C}.
\end{align}
Since $L_1$ is an infinitely divisible random variable, the
statement follows.
\end{proof}

We have seen so far, that every \lev process can be associated with
the law of an infinitely divisible distribution. The opposite, i.e.
that given any random variable $X$, whose law is infinitely
divisible, we can construct a \lev process $L=\Lt$ such that
$\mathcal{L}(L_1):=\mathcal{L}(X)$, is also true. This will be the
subject of the L\'evy-It\^o decomposition. We prepare this result
with an analysis of the jumps of a \lev \proc and the introduction
of Poisson random measures.

\section{Analysis of jumps and Poisson random measures}

The jump process \prozess[\Delta L] associated to the \lev process
$L$ is defined, for each $\ott$, via
\begin{align*}
\Delta L_t=L_t-L_{t-},
\end{align*}
where $L_{t-} = \lim_{s\uparrow t}L_s$.
The condition of stochastic continuity of a \lev
process yields immediately that for any \lev process $L$ and any
\emph{fixed} $t>0$, then $\Delta L_t=0$ a.s.; hence, a \lev \proc
has no \emph{fixed times of discontinuity}.

In general, the sum of the jumps of a \lev process does not
converge, in other words it is possible that
\begin{align*}
\sum_{s\le t}|\Delta L_s|=\infty \quad \text{a.s.}
\end{align*}
but we always have that
\begin{align*}
\sum_{s\le t}|\Delta L_s|^2<\infty \quad \text{a.s.}
\end{align*}
which allows us to handle \lev processes by martingale techniques.

A convenient tool for analyzing the jumps of a \lev process is the
\textit{random measure of jumps} of the process. Consider a set
$A\in\mathcal{B}(\R\backslash\{0\})$ such that $0\notin\overline{A}$
and let $0\le t\le T$; define the random measure of the jumps of the
process $L$ by
\begin{align}\label{rmass}
 \mu^L(\omega;t,A)
  &= \# \{0\le s\le t; \Delta L_s(\omega)\in A\}\\
  &= \sum_{s\le t}1_A(\Delta L_s(\omega));\nonumber
\end{align}
hence, the measure $\mu^L(\omega;t,A)$ \textit{counts the jumps of
the process $L$ of size in $A$ up to time $t$}. Now, we can check
that $\mu^L$ has the following properties:
\begin{align*}
\mu^L(t,A)-\mu^L(s,A)\in\sigma(\{L_u-L_v; s\le v<u\le t\})
\end{align*}
hence $\mu^L(t,A)-\mu^L(s,A)$ is independent of $\F_s$, i.e.
$\mu^L(\cdot,A)$ has \textit{independent increments}. Moreover,
$\mu^L(t,A)-\mu^L(s,A)$ equals the number of jumps of $L_{s+u}-L_s$
in $A$ for $0\le u\le t-s$; hence, by the stationarity of the
increments of $L$, we conclude:
\begin{align*}
\mathcal{L}(\mu^L(t,A)-\mu^L(s,A))=\mathcal{L}(\mu^L(t-s,A))
\end{align*}
i.e. $\mu^L(\cdot,A)$ has \textit{stationary increments}.

Hence, $\mu^L(\cdot,A)$ is a \textit{Poisson process} and $\mu^L$ is
a \textit{Poisson random measure}. The \textit{intensity} of this
Poisson process is $\nu(A)=\E[\mu^L(1,A)]$.

\begin{thm}
The set function $A\mapsto\mu^L(\omega;t,A)$ defines a
$\sigma$-finite measure on $\R\backslash\{0\}$ for each
$(\omega,t)$. The set function $\nu(A)=\E[\mu^L(1,A)]$ defines a
$\sigma$-finite measure on $\R\backslash\{0\}$.
\end{thm}
\begin{proof}
The set function $A\mapsto\mu^L(\omega;t,A)$ is simply a counting
measure on $\mathcal{B}(\R\backslash\{0\})$; hence,
\begin{align*}
\E[\mu^L(t,A)] = \int \mu^L(\omega;t,A)  \ud P(\omega)
\end{align*}
is a Borel measure on $\mathcal{B}(\R\backslash\{0\})$.
\end{proof}

\begin{defin}
The measure $\nu$ defined by
\begin{align*}
\nu(A)
 =\E[\mu^L(1,A)]
 =\E\big[\sum_{s\le1}1_A(\Delta L_s(\omega))\big]
\end{align*}
is the \textit{\lev measure} of the \lev\proc $L$.
\end{defin}

Now, using that $\mu^L(t,A)$ is a counting measure we can define an
integral with respect to the Poisson random measure $\mu^L$.
Consider a set $A\in\mathcal{B}(\R\backslash\{0\})$ such that
$0\notin\overline{A}$ and a function $f:\R\rightarrow\R$, Borel
measurable and finite on $A$. Then, the integral with respect to a
Poisson random measure is defined as follows:
\begin{align}\label{point}
\int_A f(x)\mu^L(\omega; t,\dx)
 = \sum_{s\le t} f(\Delta L_s)1_A(\Delta L_s(\omega)).
\end{align}
Note that each $\int_A f(x)\mu^L(t,\dx)$ is a real-valued random
variable and generates a c\`{a}dl\`{a}g stochastic process. We will
denote the stochastic process by $\int_0^\cdot\int_A
f(x)\mu^L(\ds,\dx)=(\int_0^t\int_A f(x)\mu^L(\ds,\dx))_\ott$.

\begin{thm}\label{poi-A-ch}
Consider a set $A\in\mathcal{B}(\R\backslash\{0\})$ with
$0\notin\overline{A}$ and a function $f:\R\rightarrow\R$, Borel
measurable and finite on $A$.
\begin{enumerate}[\textbf{A.}]
\item The process $(\int_0^t\int_A f(x)\mu^L(\ds,\dx))_{\ott}$
      is a compound Poisson process with characteristic function
\begin{align}\label{poi-cf}
 \E\Big[\exp\Big(iu\int\nolimits_0^t\!\int\nolimits_A f(x)\mu^L(\ds,\dx)\Big)\Big]
  = \exp\Big(t\int\nolimits_A(\e^{iuf(x)}-1)\nu(\dx)\Big).\vspace{-.5em}
\end{align}
\end{enumerate}
\begin{enumerate}[\textbf{B.}]
\item If $f\in L^1(A)$, then
\begin{align}
\E\Big[\int\nolimits_0^t\!\int\nolimits_A f(x)\mu^L(\ds,\dx)\Big]
  =t\int\nolimits_A f(x)\nu(\dx).
\end{align}
\end{enumerate}
\begin{enumerate}[\textbf{C.}]
\item If $f\in L^2(A)$, then
\begin{align}
\mathrm{Var}\Big(\Big|\int\nolimits_0^t\!\int\nolimits_A f(x)\mu^L(\ds,\dx)\Big|\Big)
  =t\int\nolimits_A |f(x)|^2\nu(\dx).
\end{align}
\end{enumerate}
\end{thm}
\begin{proof}[Sketch of Proof]
The structure of the proof is to start with simple functions and
pass to positive measurable functions, then take limits and use
dominated convergence; cf. Theorem 2.3.8 in \citeN{Applebaum04}.
\end{proof}

\section{The L\'evy-It\^o decomposition}

\begin{thm}
Consider a triplet $(b,c,\nu)$ where $b\in\R$, $c\in\R_{\geqslant0}$
and $\nu$ is a measure satisfying $\nu(\{0\})=0$ and
$\int_{\R}(1\wedge|x|^2)\nu(\ud x)<\infty$. Then, there exists a
probability space $(\Omega, \F, P)$ on which four independent \lev
processes exist, $L^{(1)}$, $L^{(2)}$, $L^{(3)}$ and $L^{(4)}$,
where $L^{(1)}$ is a constant drift, $L^{(2)}$ is a Brownian motion,
$L^{(3)}$ is a compound Poisson process and $L^{(4)}$ is a square
integrable (pure jump) martingale with an a.s. countable number of
jumps of magnitude less than $1$ on each finite time interval.
Taking $L=L^{(1)}+L^{(2)}+L^{(3)}+L^{(4)}$, we have that there
exists a probability space on which a \lev process $L=\Lt$ with
characteristic exponent
\begin{align}\label{exponent}
\psi(u)= iub - \frac{u^2c}{2} +
         \int_{\R}(\e^{iux}-1-iux\1_{\{|x|<1\}})\nu(\ud x)
\end{align}
for all $u\in\R$, is defined.
\end{thm}
\begin{proof}
See chapter 4 in \citeN{Sato99} or chapter 2 in \citeN{Kyprianou06}.
\end{proof}

The L\'evy-It\^o decomposition is a hard mathematical result to
prove; here, we go through some steps of the proof because it
reveals much about the structure of the paths of a \lev process. We
split the \lev exponent (\ref{exponent}) into four parts
\begin{align*}
\psi=\psi^{(1)}+\psi^{(2)}+\psi^{(3)}+\psi^{(4)}
\end{align*}
where
\begin{align*}
 \psi^{(1)}(u)&=iub,  \qquad \psi^{(2)}(u)=\frac{u^2c}{2},\\
 \psi^{(3)}(u)&=\int_{|x|\geq1}(\e^{iux}-1)\nu(\ud x),\\
 \psi^{(4)}(u)&=\int_{|x|<1}(\e^{iux}-1-iux)\nu(\ud x).
\end{align*}
The first part corresponds to a deterministic linear process (drift)
with parameter $b$, the second one to a Brownian motion with
coefficient $\sqrt{c}$ and the third part corresponds to a compound
Poisson process with arrival rate $\lambda:=\nu(\R\backslash(-1,1))$
and jump magnitude $F(\ud x):=\frac{\nu(\ud
x)}{\nu(\R\setminus(-1,1))}\1_{\{|x|\geq1\}}$.

The last part is the most difficult to handle; let $\Delta L^{(4)}$
denote the jumps of the \lev process $L^{(4)}$, that is $\Delta
L^{(4)}_t=L^{(4)}_t-L^{(4)}_{t-}$, and let $\mu^{L^{(4)}}$ denote
the random measure counting the jumps of $L^{(4)}$. Next, one
constructs a \textit{compensated} compound Poisson process
\begin{align*}
L^{(4,\epsilon)}_t &= \sum_{0\leq s\leq t}
                      \Delta L^{(4)}_s\1_{\{1>|\Delta L^{(4)}_s|>\epsilon\}}
                    - t\Big(\int_{1>|x|>\epsilon}x\nu(\ud x)\Big)\\
                   &= \int_0^t\int_{1>|x|>\epsilon}x\mu^{L^{(4)}}(\ud x, \ud s)
                    - t\Big(\int_{1>|x|>\epsilon}x\nu(\ud x)\Big)
\end{align*}
and shows that the jumps of $L^{(4)}$ form a Poisson process; using
Theorem \ref{poi-A-ch} we get that the characteristic exponent of
$L^{(4,\epsilon)}$ is
\begin{align*}
 \psi^{(4,\epsilon)}(u)&=\int_{\epsilon<|x|<1}(\e^{iux}-1-iux)\nu(\ud x).
\end{align*}
Then, there exists a \textit{\lev process} $L^{(4)}$ which is a
\textit{square integrable martingale}, such that
$L^{(4,\epsilon)}\rightarrow L^{(4)}$ uniformly on $[0,T]$ as
$\epsilon\rightarrow 0+$. Clearly, the \lev exponent of the latter
\lev\proc is $\psi^{(4)}$.

Therefore, we can decompose any \lev process into four independent
\lev processes $L=L^{(1)}+L^{(2)}+L^{(3)}+L^{(4)}$, as follows
\begin{align}\label{levy-ito}
L_t &= bt+\sqrt{c}W_t
     + \int_0^t\int_{|x|\geq1}x\mu^L(\ud s,\ud x)
     + \int_0^t\int_{|x|<1}x(\mu^L-\nu^L)(\ud s,\ud x)
\end{align}
where $\nu^L(\ud s,\ud x)=\nu(\dx)\ds$. Here $L^{(1)}$ is a constant
drift, $L^{(2)}$ a Brownian motion, $L^{(3)}$ a compound Poisson
process and $L^{(4)}$ a pure jump martingale. This result is the
celebrated \emph{L\'evy-It\^o decomposition} of a \lev process.

\section{The L\'evy measure, path and moment properties}

The \lev measure $\nu$ is a measure on $\R$ that satisfies
\begin{eqnarray}\label{levy-measure}
  \nu(\{0\})=0
  &\qquad  \mathrm{and} &\qquad
  \int_{\R}(1\wedge|x|^2)\nu(\dx)<\infty.
\end{eqnarray}
Intuitively speaking, the \lev measure describes \emph{the expected
number of jumps of a certain height in a time interval of length 1}.
The \lev measure has no mass at the origin, while singularities
(i.e. infinitely many jumps) can occur around the origin (i.e. small
jumps). Moreover, the mass away from the origin is bounded (i.e.
only a finite number of big jumps can occur).

Recall the example of the \lev jump-diffusion; the \lev measure is
$\nu(\ud x)=\lambda\times F(\ud x)$; from that we can deduce that
the expected number of jumps is $\lambda$ and the jump size is
distributed according to $F$.

More generally, if $\nu$ is a finite measure, i.e.
$\lambda:=\nu(\R)=\int_{\R}\nu(\dx)<\infty$, then we can define
$F(\dx):=\frac{\nu(\dx)}{\lambda}$, which is a probability measure.
Thus, $\lambda$ is the expected number of jumps and $F(\dx)$ the
distribution of the jump size $x$. If $\nu(\R)=\infty$, then an
infinite number of (small) jumps is expected.

\begin{figure}
 \begin{center}
   \includegraphics[width=6.05cm,keepaspectratio=true]{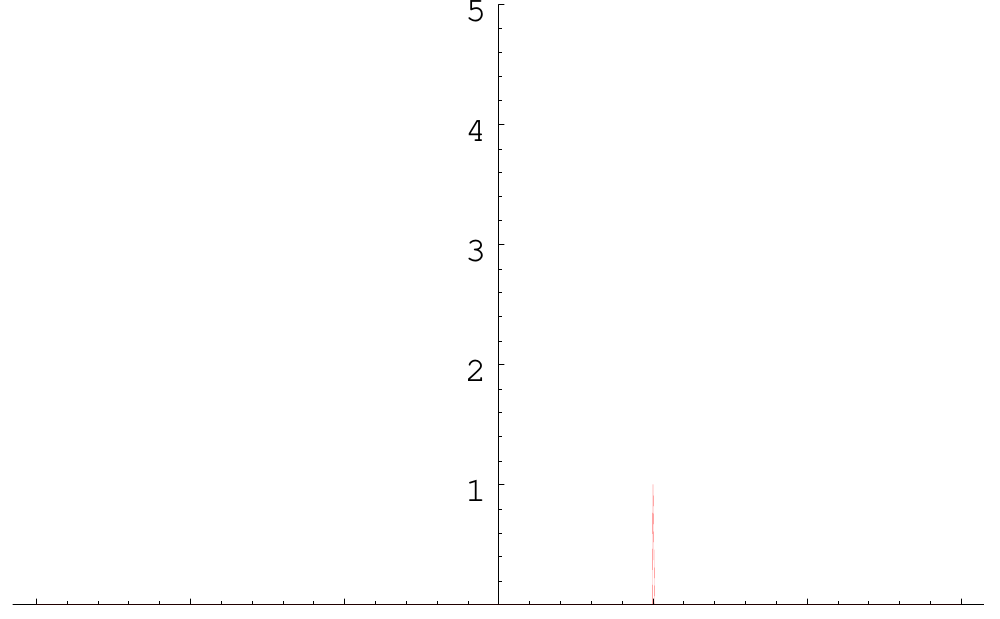}
   \includegraphics[width=5.87cm,keepaspectratio=true]{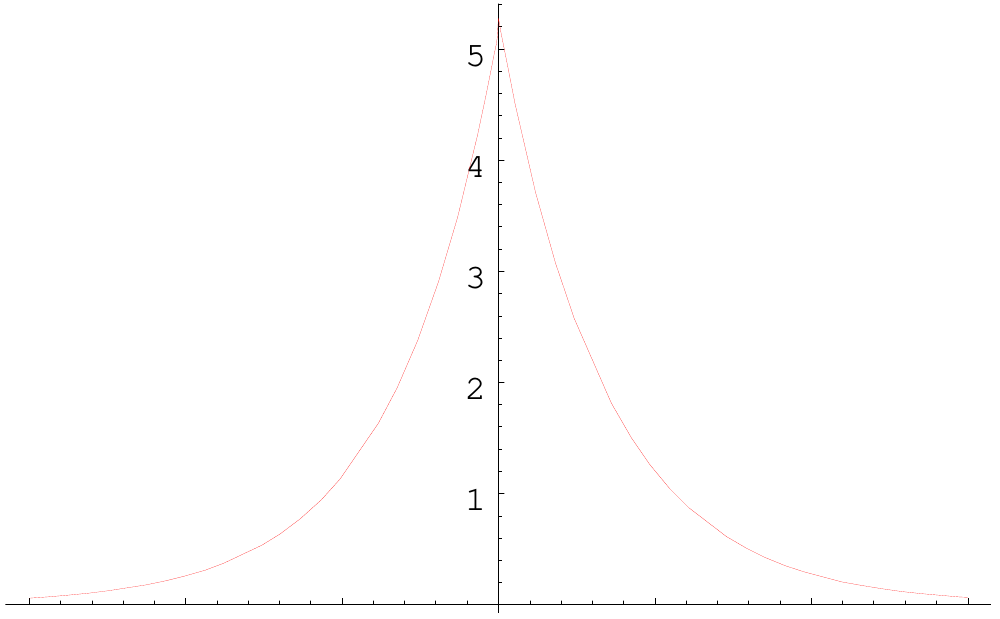}
    \caption{The distribution function of the \lev measure of the standard
             Poisson process (left) and the density of the \lev measure of
             a compound Poisson process with double-exponentially distributed
             jumps.}
    \label{P-DBE}
 \end{center}
\end{figure}

\begin{figure}
 \begin{center}
  \includegraphics[width=5.9cm,keepaspectratio=true]{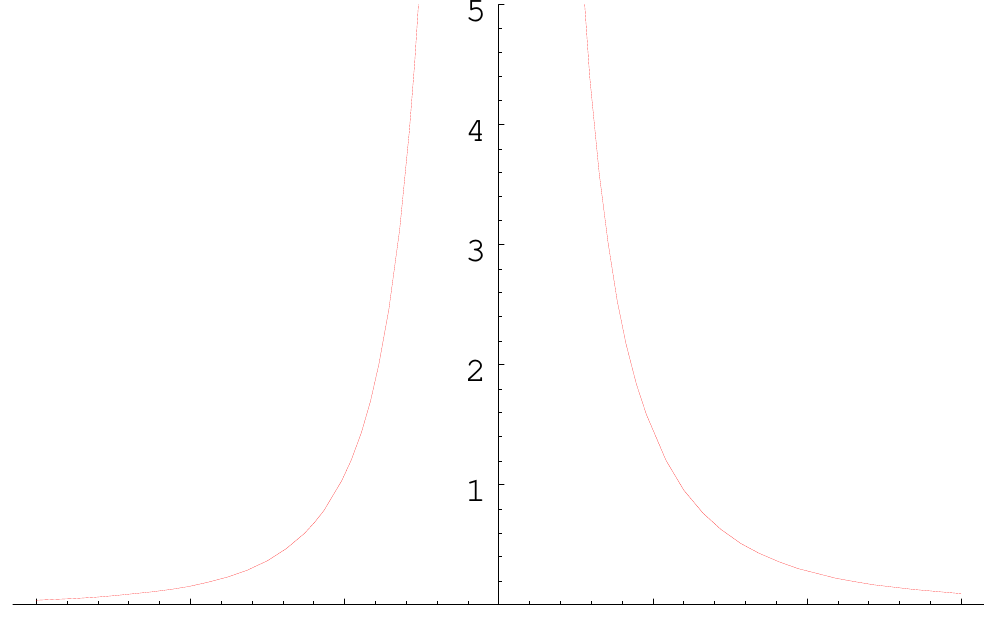}
  \includegraphics[width=5.97cm,keepaspectratio=true]{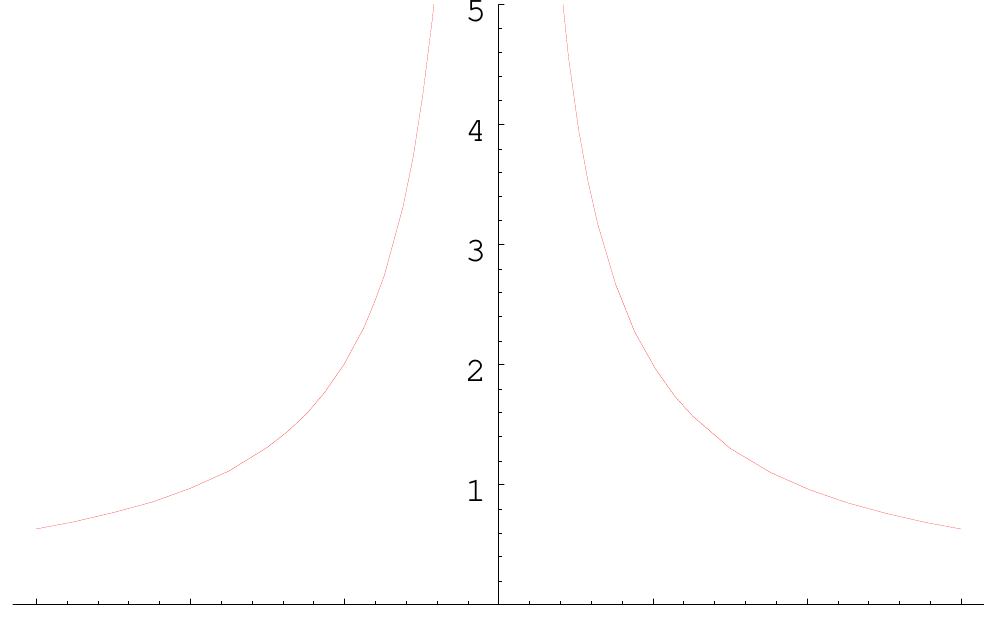}
   \caption{The density of the \lev measure of an NIG (left) and an
            $\alpha$-stable process.}
   \label{NIG-ST}
 \end{center}
\end{figure}

The \lev measure is responsible for the \emph{richness} of the class
of \lev processes and carries useful information about the structure
of the process. Path properties can be read from the \lev measure:
for example, Figures \ref{P-DBE} and \ref{NIG-ST} reveal that the
compound Poisson process has a finite number of jumps on every time
interval, while the NIG and $\alpha$-stable processes have an
infinite one; we then speak of an \emph{infinite activity} \lev
process.

\begin{prop}\label{activity}
Let $L$ be a \lev process with triplet $(b,c,\nu)$.
\begin{enumerate}
 \item If $\nu(\R)<\ap$, then almost all paths of $L$ have a finite
       number of jumps on every compact interval. In that case, the
       \lev process has \emph{finite activity}.
 \item If $\nu(\R)=\ap$, then  almost all paths of $L$ have an
       infinite number of jumps on every compact interval. In that
       case, the \lev process has \emph{infinite activity}.
\end{enumerate}
\end{prop}
\begin{proof}
See Theorem 21.3 in \citeN{Sato99}.
\end{proof}

Whether a \lev process has finite variation or not also depends on
the \lev measure (and on the presence or absence of a Brownian
part).

\begin{prop}\label{variation}
Let $L$ be a \lev process with triplet $(b,c,\nu)$.
\begin{enumerate}
 \item If $c=0$ and $\int_{|x|\leq1}|x|\nu(\dx)<\infty$, then almost all paths of $L$ have
       \emph{finite variation}.
 \item If $c\neq0$ or $\int_{|x|\leq1}|x|\nu(\dx)=\infty$, then almost all paths of $L$ have
       \emph{infinite variation}.
\end{enumerate}
\end{prop}
\begin{proof}
See Theorem 21.9 in \citeN{Sato99}.
\end{proof}

The different functions a \lev measure has to integrate in order to
have finite activity or variation, are graphically exhibited in
Figure \ref{LMFFV}. The compound Poisson process has finite measure,
\emph{hence} it has finite variation as well; on the contrary, the
NIG \lev process has an infinite measure \emph{and} has infinite
variation. In addition, the CGMY \lev process for $0<Y<1$ has
infinite activity, \textit{but} the paths have finite variation.

\begin{figure}
 \begin{center}
  \includegraphics[width=8cm,keepaspectratio=true]{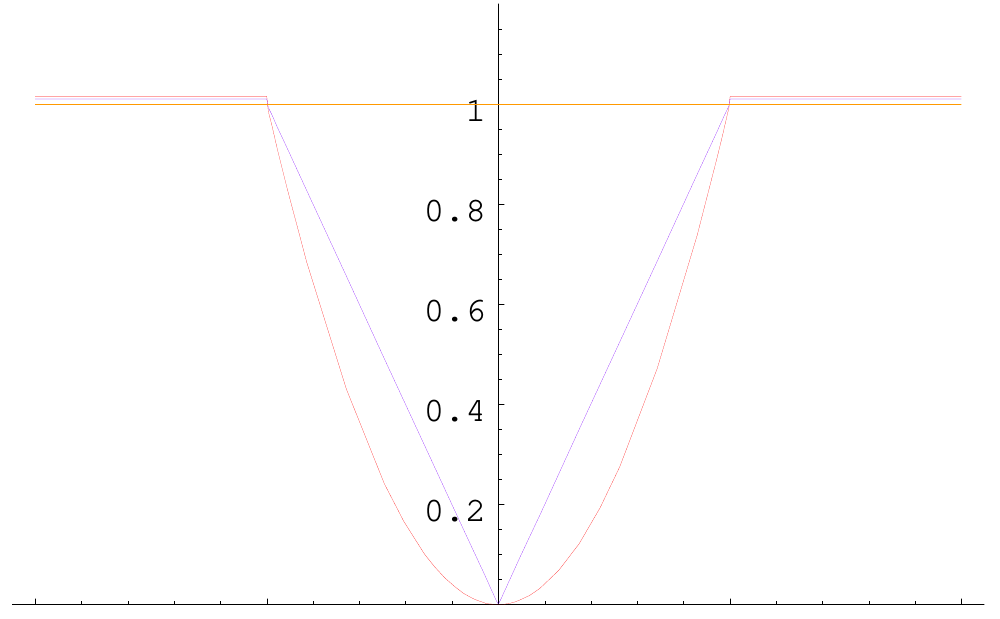}
  \caption{The \lev measure must integrate $|x|^2\wedge1$ (red line);
           it has finite variation if it integrates  $|x|\wedge1$
           (blue line); it is finite if it integrates $1$ (orange line).}
  \label{LMFFV}
 \end{center}
\end{figure}

The \lev measure also carries information about the finiteness of
the \emph{moments} of a \lev process. This is particularly useful
information in mathematical finance, related to the existence of a
\emph{martingale measure}.

The finiteness of the moments of a \lev process is related to the
finiteness of an integral over the \lev measure (more precisely, the
restriction of the L\'evy measure to jumps larger than $1$ in
absolute value, i.e. \textit{big} jumps).

\begin{prop}\label{moments}
Let $L$ be a \lev process with triplet $(b,c,\nu)$. Then
\begin{enumerate}
\item $L_t$ has \emph{finite $p$-th moment} for $p\in\R_{\geqslant0}$ $(\E|L_t|^p<\infty)$
      if and only if $\int_{|x|\geq1}|x|^p\nu(\dx)<\ap$.
\item $L_t$ has \emph{finite $p$-th exponential moment} for $p\in\R$
      $(\E[\e^{pL_t}]<\infty)$ if and only if $\int_{|x|\geq1}\e^{px}\nu(\dx)<\ap$.
\end{enumerate}
\end{prop}
\begin{proof}
The proof of these results can be found in Theorem 25.3 in
\citeN{Sato99}. Actually, the conclusion of this theorem holds for
the general class of \emph{submultiplicative} functions (cf.
Definition 25.1 in \citeNP{Sato99}), which contains $\exp(px)$ and
$|x|^p\vee1$ as special cases.
\end{proof}

In order to gain some understanding of this result and because it
blends beautifully with the L\'evy-It\^o decomposition, we will give
a rough proof of the sufficiency for the second statement (inspired
by \citeNP{Kyprianou06}).

Recall from the L\'evy-It\^o decomposition, that the characteristic
exponent of a \lev process was split into four independent parts,
the third of which is a compound Poisson process with arrival rate
$\lambda:=\nu(\R\backslash(-1,1))$ and jump magnitude $F(\ud
x):=\frac{\nu(\ud x)}{\nu(\R\backslash(-1,1))}\1_{\{|x|\geq1\}}$.
Finiteness of $\E[\e^{pL_t}]$ implies finiteness of
$\E[\e^{pL^{(3)}_t}]$, where
\begin{align*}
\E[\e^{pL^{(3)}_t}]
  &= \e^{-\lambda t}\sum_{k\geq0}\frac{(\lambda t)^{k}}{k!}
        \left( \int_{\R}\e^{px}F(\ud x)\right)^k\\
  &= \e^{-\lambda t}\sum_{k\geq0}\frac{t^{k}}{k!}
        \left( \int_{\R}\e^{px}\1_{\{|x|\geq1\}}\nu(\ud x)\right)^k.
\end{align*}
Since all the summands must be finite, the one corresponding to $k=1$
must also be finite, therefore
\begin{align*}
\e^{-\lambda t}\int_{\R}\e^{px}\1_{\{|x|\geq1\}}\nu(\ud x)<\ap
 \Longrightarrow
\int_{|x|\geq1}\e^{px}\nu(\ud x)<\ap.
\end{align*}

\begin{figure}
 \begin{center}
  \includegraphics[width=8cm,keepaspectratio=true]{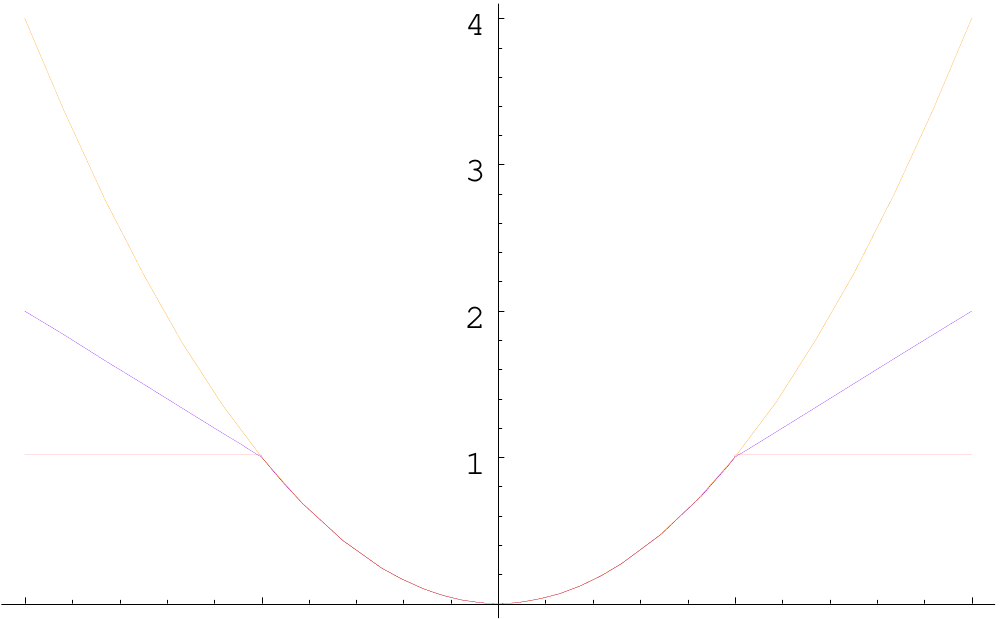}
  \caption{A \lev process has first moment if the \lev measure integrates $|x|$
           for $|x|\geq1$ (blue line) and second moment if it integrates $x^2$
           for $|x|\geq1$ (orange line).}
  \label{LMM}
 \end{center}
\end{figure}

The graphical representation of the functions the \lev measure must
integrate so that a \lev process has finite moments is given in
Figure \ref{LMM}. The NIG process possesses moments of all order,
while the $\alpha$-stable does not; one can already observe in
Figure \ref{NIG-ST} that the tails of the \lev measure of the
$\alpha$-stable are much heavier than the tails of the NIG.

\begin{rem}
As can be observed from Propositions \ref{activity}, \ref{variation}
and \ref{moments}, the \textit{variation} of a \lev process depends
on the \textit{small jumps} (and the Brownian motion), the
\textit{moment} properties depend on the \textit{big jumps}, while
the \textit{activity} of a \lev process depends on \textit{all} the
jumps of the process.
\end{rem}

\section{Some classes of particular interest}

We already know that a Brownian motion, a (compound) Poisson process
and a \lev jump-diffusion are \lev processes, their L\'evy-It\^o
decomposition and their characteristic functions. Here, we present
some further subclasses of \lev processes that are of special
interest.

\subsection{Subordinator}\label{subordinator}

A \emph{subordinator} is an a.s. increasing (in $t$) \lev process.
Equivalently, for $L$ to be a subordinator, the triplet must satisfy
$\nu(-\infty,0)=0$, $c=0$, $\int_{(0,1)}x\nu(\ud x)<\infty$ and
$\gamma=b-\int_{(0,1)}x\nu(\ud x)>0$.

The L\'evy-It\^o decomposition of a subordinator is
\begin{equation}
L_t = \gamma t + \int_0^t\int_{(0,\infty)}x \mu^L(\ud s, \ud x)
\end{equation}
and the L\'evy-Khintchine formula takes the form
\begin{equation}
\E[\e^{iuL_t}]
 = \exp\Big[t\big(iu\gamma +\int_{(0,\infty)}(\e^{iux}-1)\nu(\ud x)\big)\Big].
\end{equation}
Two examples of subordinators are the Poisson and the inverse
Gaussian process, cf. Figures \ref{NIG-GIG} and
\ref{poisson-compoundpoisson}.

\subsection{Jumps of finite variation}

A \lev process has jumps of finite variation if and only if
$\int_{|x|\leq1}|x|\nu(\ud x)<\infty$. In this case, the
L\'evy-It\^o decomposition of $L$ resumes the form
\begin{equation}
L_t = \gamma t + \sqrt{c}W_t + \int_0^t\int_{\R}x \mu^L(\ud s, \ud x)
\end{equation}
and the L\'evy-Khintchine formula takes the form
\begin{equation}
\E[\e^{iuL_t}] = \exp\Big[t\big(iu\gamma - \frac{u^2c}{2} +
                       \int_{\R}(\e^{iux}-1)\nu(\ud x)\big)\Big],
\end{equation}
where $\gamma$ is defined similarly to subsection \ref{subordinator}.

Moreover, if $\nu([-1,1])<\infty$, which means that
$\nu(\R)<\infty$, then the jumps of $L$ correspond to a
\emph{compound Poisson} process.

\subsection{Spectrally one-sided}

A \lev processes is called \textit{spectrally negative} if
$\nu(0,\infty)=0$. The L\'evy-It\^o decomposition of a spectrally
negative \lev process has the form
\begin{align}
L_t &= bt + \sqrt{c}W_t + \int_0^t\!\int_{x<-1}x \mu^L(\ud s, \ud x)
          + \int_0^t\!\!\int_{-1<x<0}\!\!\!\!x(\mu^L-\nu^L)(\ud s,\ud x)
\end{align}
and the L\'evy-Khintchine formula takes the form
\begin{equation}
\E[\e^{iuL_t}] = \exp\Big[t\big(iub - \frac{u^2c}{2} +
                       \int_{(-\infty,0)}(\e^{iux}-1-iu1_{\{x>-1\}})\nu(\ud x)\Big].
\end{equation}

Similarly, a \lev processes is called \textit{spectrally positive}
if $-L$ is spectrally negative.

\subsection{Finite first moment}

As we have seen already, a \lev process has finite first moment if
and only if $\int_{|x|\geq1}|x|\nu(\ud x)<\infty$. Therefore, we can
also compensate the big jumps to form a martingale, hence the
L\'evy-It\^o decomposition of $L$ resumes the form
\begin{equation}\label{LI-SSMG}
L_t = b't + \sqrt{c}W_t
    + \int_0^t\int_{\R}x (\mu^L-\nu^L)(\ud s,\ud x)
\end{equation}
and the L\'evy-Khintchine formula takes the form
\begin{equation}
\E[\e^{iuL_t}] = \exp\Big[t\big(iub' - \frac{u^2c}{2} +
                       \int_{\R}(\e^{iux}-1-iux)\nu(\ud x)\big)\Big],
\end{equation}
where $b'=b+\int_{|x|\geq1}x\nu(\ud x)$.

\begin{rem}[\textbf{Assumption} ($\mathbb M$)]\label{EM}
For the remaining parts we will work only with \lev process that
have finite first moment. We will refer to them as \lev processes
that satisfy \textit{Assumption} ($\mathbb M$). For the sake of
simplicity, we suppress the notation $b'$ and write $b$ instead.
\end{rem}

\begin{figure}
\begin{center}
 \includegraphics[width=6.cm,keepaspectratio=true]{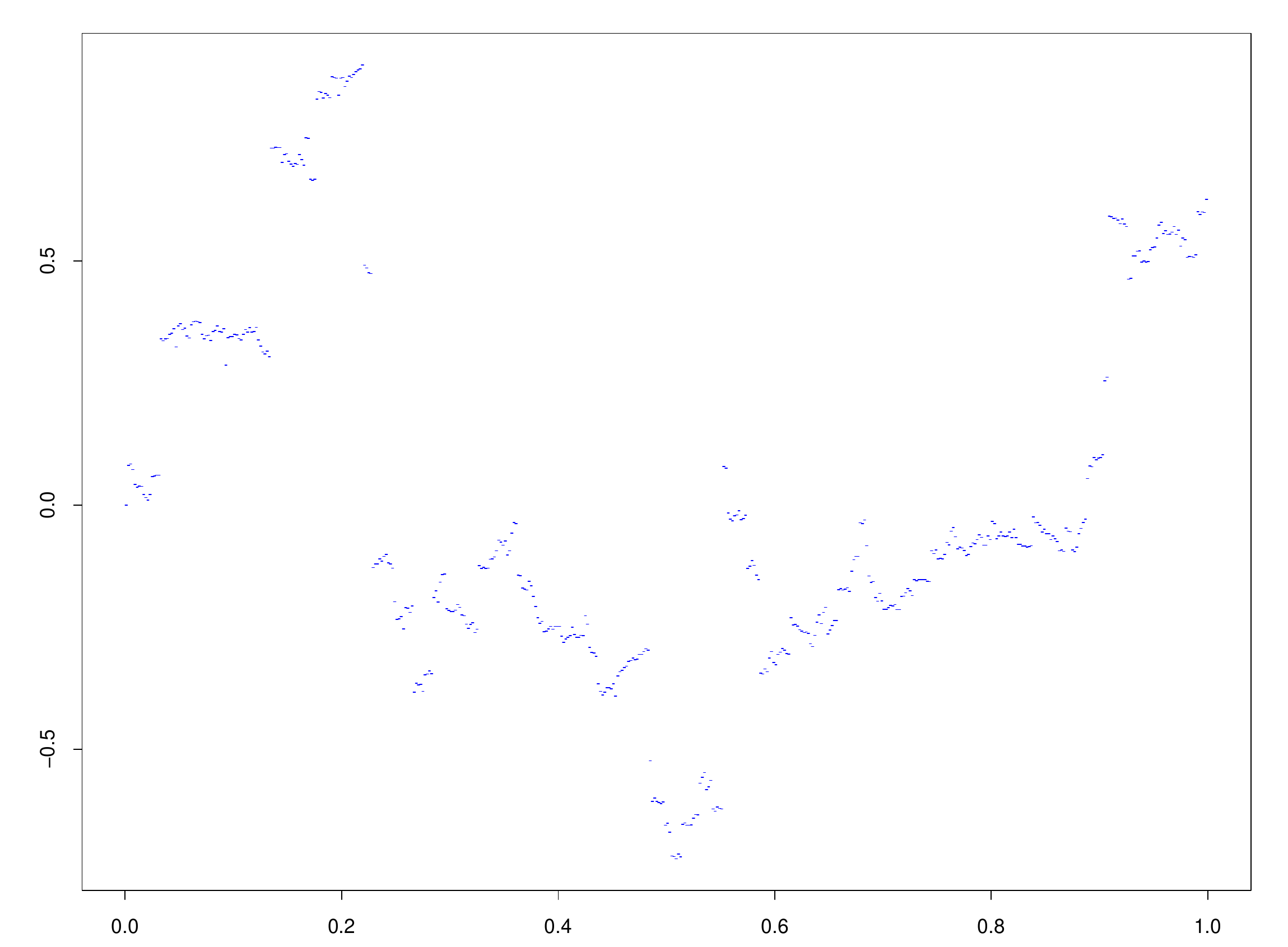}
 \includegraphics[width=6.cm,keepaspectratio=true]{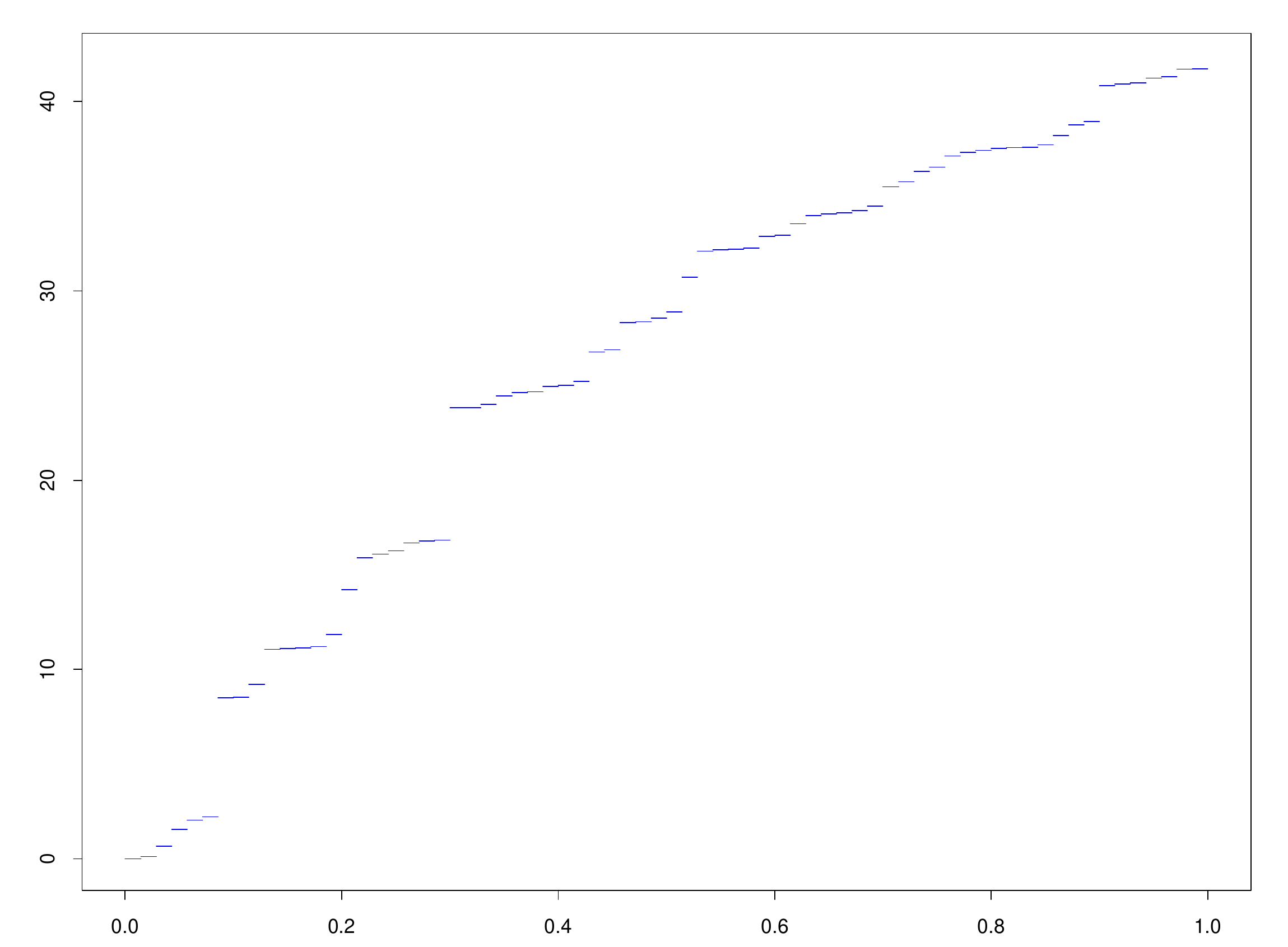}
  \caption{Simulated path of a normal inverse Gaussian (left) and an
           inverse Gaussian process.}
   \label{NIG-GIG}
\end{center}
\end{figure}

\section{Elements from semimartingale theory}

A \emph{semimartingale} is a stochastic process \prozess[X] which
admits the decomposition
\begin{align}\label{gsmmg}
X=X_0+M+A
\end{align}
where $X_0$ is finite and $\F_0$-measurable, $M$ is a local
martingale with $M_0=0$ and $A$ is a finite variation process with
$A_0=0$. $X$ is a \emph{special semimartingale} if $A$ is
\emph{predictable}.

Every special semimartingale $X$ admits the following, so-called,
\emph{canonical decomposition}
\begin{align}\label{smmg}
X &= X_0 + B + X^c + x*(\mu^X-\nu^X).
\end{align}
Here $X^c$ is the \emph{continuous martingale} part of $X$ and
$x*(\mu^X-\nu^X)$ is the \emph{purely discontinuous martingale} part
of $X$.  $\mu^X$ is called the \emph{random measure of jumps} of
$X$; it counts the number of jumps of specific size that occur in a
time interval of specific length. $\nu^X$ is called the
\emph{compensator} of $\mu^X$; for a detailed account, we refer to
\citeN[Chapter II]{JacodShiryaev03}.

\begin{rem}
Note that $W*\mu$, for $W=W(\omega; s,x)$ and the integer-valued
measure $\mu=\mu(\omega;\ud t, \ud x)$, $t\in[0,T]$, $x\in E$,
denotes the \emph{integral process}
\begin{displaymath}
  \int_0^\cdot \int_E  W(\omega;t,x)\mu(\omega;\ud t,\ud x).
\end{displaymath}
Consider a predictable function $W:\Omega\times[0,T]\times E
\rightarrow\R$ in $G_\text{loc}(\mu)$; then $W*(\mu-\nu)$ denotes
the \emph{stochastic integral}
\begin{displaymath}
  \int_0^\cdot \int_E  W(\omega;t,x)(\mu-\nu)(\omega;\ud t,\ud x).
\end{displaymath}
\end{rem}

Now, recalling the L\'evy-It\^o decomposition \eqref{LI-SSMG} and
comparing it to \eqref{smmg}, we can easily deduce that a \lev
process with triplet ($b,c,\nu$) which satisfies Assumption
($\mathbb{M}$), has the following \textit{canonical decomposition}
\begin{align}\label{ssmmg-lp}
L_t&=bt+\sqrt{c}W_t+\int_0^t\int_{\R}x(\mu^L-\nu^L)(\dsdx),
\end{align}
where
\begin{align*}
\int_0^t\int_{\R}x\mu^L(\dsdx)=\sum_{0\leq s\leq t}\Delta L_s
\end{align*}
and
\begin{align*}
\E\Big[\int_0^t\int_{\R}x\mu^L(\dsdx)\Big] = \int_0^t\int_{\R} x\nu^L(\dsdx)
                                           = t \int_{\R} x\nu(\dx).
\end{align*}
Therefore, a \lev process that satisfies Assumption ($\mathbb{M}$)
is a \textit{special semimartingale} where the continuous martingale
part is a \textit{Brownian motion} with coefficient $\sqrt{c}$ and
the random measure of the jumps is a \textit{Poisson random
measure}. The compensator $\nu^L$ of the Poisson random measure
$\mu^L$ is a product measure of the \lev measure with the Lebesgue
measure, i.e. $\nu^L=\nu\otimes\llambda$; one then also writes
$\nu^L(\dsdx)=\nu(\dx)\ud s$.

We denote the \emph{continuous martingale} part of $L$ by $L^c$ and
the \emph{purely discontinuous martingale} part of $L$ by $L^d$,
i.e.
\begin{eqnarray}
L_t^c = \sqrt{c}W_t
 && \text{  and  } \qquad
L_t^d = \int_0^t\int_{\R}x(\mu^L-\nu^L)(\dsdx).
\end{eqnarray}

\begin{rem}
Every \lev process is also a semimartingale; this follows easily
from \eqref{gsmmg} and the L\'evy--It\^o decomposition of a \lev
process. Every \lev process with finite first moment (i.e. that
satisfies Assumption ($\mathbb{M}$)) is also a \textit{special}
semimartingale; conversely, every \lev process that is a special
semimartingale, has a finite first moment. This is the subject of
the next result.
\end{rem}

\begin{lem}\label{ssmg}
Let $L$ be a \lev process with triplet ($b,c,\nu$). The following
conditions are equivalent
\begin{enumerate}
\item $L$ is a \emph{special} semimartingale,
\item $\int_{\R}(|x|\wedge|x|^2)\nu(\dx)<\infty$,
\item $\int_{\R}|x|1_{\{|x|\geq1\}}\nu(\dx)<\infty$.
\end{enumerate}
\end{lem}
\begin{proof}
From Lemma 2.8 in \citeN{KallsenShiryaev02} we have that, a \lev
process (semimartingale) is special if and only if the compensator
of its jump measure satisfies
\begin{align*}
\int_{0}^{\cdot}\int_{\R}(|x|\wedge|x|^2)\nu^L(\dsdx)\in \mathcal{V}.
\end{align*}
For a fixed $t\in\R$, we get
\begin{align*}
\int_{0}^{t}\int_{\R}(|x|\wedge|x|^2)\nu^L(\dsdx)
 &= \int_{0}^{t}\int_{\R}(|x|\wedge|x|^2)\nu(\dx)\ds\\
 &= t\cdot\int_{\R}(|x|\wedge|x|^2)\nu(\dx)
\end{align*}
and the last expression is an element of $\mathcal{V}$ if
and only if
$$\int_{\R}(|x|\wedge|x|^2)\nu(\dx)<\infty;$$
this settles $(1)\Leftrightarrow(2)$. The equivalence
$(2)\Leftrightarrow(3)$ follows from the properties of the \lev
measure, namely that $\int_{|x|<1}|x|^2\nu(\dx)<\infty$, cf.
\eqref{levy-measure}.
\end{proof}

\section{Martingales and L\'evy processes}

We give a condition for a \lev process to be a martingale and
discuss when the exponential of a \lev process is a martingale.

\begin{prop}
Let $L=\Lt$ be a \lev process with \lev triplet $(b,c,\nu)$ and
assume that $\E[|L_t|]<\infty$, i.e. Assumption $(\mathbb M)$ holds.
$L$ is a martingale if and only if $b=0$. Similarly, $L$ is a
submartingale if $b>0$ and a supermartingale if $b<0$.
\end{prop}
\begin{proof}
The assertion follows immediately from the decomposition of a \lev
process with finite first moment into a finite variation process, a
continuous martingale and a pure-jump martingale, cf. equation
\eqref{ssmmg-lp}.
\end{proof}

\begin{prop}\label{exp-mg}
Let $L=\Lt$ be a \lev process with \lev triplet $(b,c,\nu)$, assume
that $\int_{|x|\geq1}\e^{ux}\nu(\dx)<\infty$, for $u\in\R$ and
denote by $\kappa$ the cumulant of $L_1$, i.e. $\kappa(u)=\log
\E[\e^{uL_1}]$. The process $M=(M_t)_\ott$, defined via
\begin{align*}
M_t = \frac{\e^{uL_t}}{\e^{t\kappa(u)}}
\end{align*}
is a martingale.
\end{prop}
\begin{proof}
Applying Proposition \ref{moments}, we get that
$\E[\e^{uL_t}]=\e^{t\kappa(u)}<\infty$, for all $\ott$. Now, for
$0\leq s\leq t$, we can re-write $M$ as
\begin{align*}
M_t = \frac{\e^{uL_s}}{\e^{s\kappa(u)}}\frac{\e^{u(L_t-L_s)}}{\e^{(t-s)\kappa(u)}}
    = M_s\frac{\e^{u(L_t-L_s)}}{\e^{(t-s)\kappa(u)}}.
\end{align*}
Using the fact that a \lev process has stationary and independent
increments, we can conclude
\begin{align*}
\E\Big[M_t\Big|\F_s\Big]
 &= M_s \E\Big[\frac{\e^{u(L_t-L_s)}}{\e^{(t-s)\kappa(u)}}\Big|\F_s\Big]
  = M_s \e^{(t-s)\kappa(u)}\e^{-(t-s)\kappa(u)}\\
 &= M_s. \qedhere
\end{align*}
\end{proof}

The \emph{stochastic exponential} $\mathcal{E}(L)$ of a \lev process
$L=\Lt$ is the solution $Z$ of the stochastic differential equation
\begin{equation}
\ud Z_t = Z_{t-}\ud L_t, \qquad Z_0=1,
\end{equation}
also written as
\begin{equation}
 Z = 1+Z_-\cdot L,
\end{equation}
where $F\cdot Y$ means the stochastic integral $\int_0^\cdot F_s\ud
Y_s$. The stochastic exponential is defined as
\begin{align}
\mathcal{E}(L)_t = \exp\left(L_t -\frac{1}{2}\langle L^c\rangle_t\right)
                   \prod_{0\leq s\leq t}\Big(1+\Delta L_s\Big)\e^{-\Delta L_s}.
\end{align}

\begin{rem}
The stochastic exponential of a \lev process that is a
\textit{martingale} is a \textit{local martingale} (cf.
\citeNP[Theorem I.4.61]{JacodShiryaev03}) and indeed a (true)
\textit{martingale} when working in a finite time horizon (cf.
\citeNP[Lemma 4.4]{Kallsen00}).
\end{rem}

The converse of the stochastic exponential is the \emph{stochastic
logarithm}, denoted $\Log X$; for a process $X=(X_t)_{0\le t\le T}$,
the stochastic logarithm is the solution of the stochastic
differential equation:
\begin{equation}  \label{eq:8}
  \Log X_t=\int_0^t \frac{\ud X_s}{X_{s-}},
\end{equation}
also written as
\begin{equation}
 \Log X = \frac{1}{X_-} \cdot X.
\end{equation}
Now, if $X$ is a positive process with $X_0=1$ we have for $\Log X$
\begin{equation} \label{eq:11a}
  \Log X = \log X+\frac{1}{2X^2_-} \cdot \langle X^c\rangle
         - \sum_{0\le s\le \cdot}
           \bigg(\log\Big(1+\frac{\Delta X_s}{X_{s-}}\Big)-\frac{\Delta X_s}{X_{s-}}\bigg);
\end{equation}
for more details see \citeN{KallsenShiryaev02} or
Jacod and Shiryaev \citeyear{JacodShiryaev03}.

\section{It\^o's formula}

We state a version of It\^o's formula directly for semimartingales,
since this is the natural framework to work into.

\begin{lem}
Let \prozess[X] be a real-valued semimartingale and $f$ a class
$C^2$ function on $\R$. Then, $f(X)$ is a semimartingale and we have
\begin{align}
f(X_t) &= f(X_0) + \int_0^t f'(X_{s-})\ud X_s
        + \frac12 \int_0^t f''(X_{s-})\ud \langle X^c\rangle_s\\
       &\quad \nonumber
        + \sum_{0\leq s\leq t} \Big(f(X_s)-f(X_{s-})-f'(X_{s-})\Delta X_s\Big),
\end{align}
for all $t\in[0,T]$; alternatively, making use of the random measure
of jumps, we have
\begin{align}
f(X_t) &= f(X_0) + \int_0^t f'(X_{s-})\ud X_s
        + \frac12 \int_0^t f''(X_{s-})\ud \langle X^c\rangle_s\\
       &\quad \nonumber
        + \int_0^t\int_\R\Big(f(X_{s-}+x)-f(X_{s-})-f'(X_{s-})x\Big)\mu^X(\dsdx).
\end{align}
\end{lem}
\begin{proof}
See Theorem I.4.57 in \citeN{JacodShiryaev03}.
\end{proof}

\begin{rem}
An interesting account (and proof) of It\^o's formula for L\'evy
processes of \textit{finite variation} can be found in
\citeN[Chapter 4]{Kyprianou06}.
\end{rem}

\begin{lem}[Integration by parts]
Let  $X,Y$ be semimartingales. Then $XY$ is also a semimartingale
and
\begin{align}
 XY = \int X_-\ud Y + \int Y_-\ud X + [X,Y],
\end{align}
where the quadratic covariation of $X$ and $Y$ is given by
\begin{align}
 [X,Y] = \langle X^c,Y^c\rangle + \sum_{s\le\cdot}\Delta X_s \Delta Y_s.
\end{align}
\end{lem}
\begin{proof}
See Corollary II.6.2 in \citeN{Protter04} and Theorem I.4.52 in
Jacod and Shiryaev \citeyear{JacodShiryaev03}.
\end{proof}

As a simple application of It\^o's formula for L\'evy processes, we
will work out the dynamics of the stochastic logarithm of a \lev
process.

Let \prozess be a \lev process with triplet ($b,c,\nu$) and $L_0=1$.
Consider the $C^2$ function $f:\R\rightarrow\R$ with $f(x)=\log|x|$;
then, $f'(x)=\frac{1}{x}$ and $f''(x)=-\frac{1}{x^2}$. Applying
It\^o's formula to $f(L)=\log|L|$, we get
\begin{align*}
\log|L_t| &= \log|L_0| + \int_0^t\frac1{L_{s-}}\ud L_s
           - \frac12\int_0^t\frac1{L^2_{s-}}\ud \langle L^c\rangle_s\\
 &\qquad   + \sum_{0\leq s\leq t}\Big(\log|L_s|-\log|L_{s-}|-\frac1{L_{s-}}\Delta L_s \Big)\\
\Leftrightarrow
\Log L_t &= \log|L_t| + \frac12\int_0^t\frac{\ud \langle L^c\rangle_s}{L^2_{s-}}
          - \sum_{0\leq s\leq t}\Big(\log\Big|\frac{L_s}{L_{s-}}\Big|-\frac{\Delta L_s}{L_{s-}} \Big).
\end{align*}
Now, making again use of the random measure of jumps of the process
$L$ and using also that $\ud \langle L^c\rangle_s=\ud \langle
\sqrt{c}W\rangle_s=c\ds$, we can conclude that
\begin{align*}
\Log L_t &= \log|L_t| + \frac{c}2 \int_0^t \frac{\ds}{L^2_{s-}}
          - \int_0^t\int_\R\Big(\log\Big|1+\frac{x}{L_{s-}}\Big|-\frac{x}{L_{s-}}\Big)\mu^L(\dsdx).
\end{align*}

\section{Girsanov's theorem}

We will describe a special case of \emph{Girsanov's theorem for
semimartingales}, where a \lev process remains a \emph{process with
independent increments} (PII) under the new measure. Here we will
restrict ourselves to a \textit{finite} time horizon, i.e.
$T\in[0,\ap)$.

Let $P$ and $\bar{P}$ be probability measures defined on the
filtered probability space ($\Omega, \F, \bF$). Two measures $P$ and
$\bar{P}$ are \emph{equivalent}, if $P(A)=0\Leftrightarrow
\bar{P}(A)=0$, for all $A\in\F$, and then one writes $P\sim\bar{P}$.

Given two equivalent measures $P$ and $\bar{P}$, there exists a
unique, positive, $P$-martingale $Z=(Z_t)_\ott$ such that
$Z_t=\E\big[\frac{\ud\bar{P}}{\ud P}\big|\F_t\big]$, $\forall\;
\ott$. $Z$ is called the \emph{density process} of $\bar{P}$ with
respect to $P$.

Conversely, given a measure $P$ and a positive $P$-martingale
$Z=(Z_t)_\ott$, one can define a measure $\bar{P}$ on
($\Omega,\F,\bF$) equivalent to $P$, using the Radon-Nikodym
derivative $\E\big[\frac{\ud \bar{P}}{\ud P}\big|\F_T\big]=Z_T$.

\begin{thm}\label{girsanov}
Let \prozess be a \lev process with triplet $(b,c,\nu)$ under $P$,
that satisfies Assumption $(\mathbb{M})$, cf. Remark \ref{EM}. Then,
$L$ has the canonical decomposition
\begin{align}\label{can-deco}
L_t=bt +\sqrt{c}W_t +\int_0^t\int_{\R}x(\mu^L-\nu^L)(\ud s,\ud x).
\end{align}
\begin{enumerate}[\emph{\textbf{(A1):}}]
\item Assume that $P\sim \bar{P}$ with density process $Z$. Then,
      there exist a deterministic process $\beta$ and a measurable
      non-negative deterministic process $Y$, satisfying
\begin{align}\label{Girsanov-moment}
 \int_0^t\int_{\R}|x\big(Y(s,x)-1\big)|\nu(\ud x)\ud s<\infty,
\end{align}
and
\begin{align*}
 \int_0^t \big( c\cdot\beta^2_s\big)\ud s<\infty,
\end{align*}
$\P$-a.s. for $\ott$; they are defined by the following formulae:
\begin{align}
\langle Z^c,L^c \rangle
 = \int_0^\cdot (c \cdot\beta_s\cdot Z_{s-})\ud s
\end{align}
and
\begin{align}\label{girsanov-Y}
 Y = M_{\mu^L}^P\left(\frac{Z}{Z_-}\Big|\widetilde{\mathcal{P}}\right).
\end{align}
\end{enumerate}
\begin{enumerate}[\emph{\textbf{(A2):}}]
\item Conversely, if $Z$ is a positive martingale of the form
\begin{align}\label{explicit-density}
 Z &= 
    \exp\Big[ \int_0^\cdot \beta_s\sqrt{c}\ud W_s
                    -\frac{1}{2}\int_0^\cdot \beta^2_s c \ud s\\
 &\qquad\qquad
                   +\int_0^\cdot\int_{\R} (Y(s,x)-1)(\mu^L-\nu^L)(\dsdx)\nonumber\\
 &\qquad\qquad
                   -\int_0^\cdot\int_{\R} (Y(s,x)-1-\ln(Y(s,x)))\mu^L(\dsdx)\Big].\nonumber
\end{align}
      then it defines a probability
      measure $\bar{P}$ on $(\Omega,\F,\bF)$, such that
      $P\sim\bar{P}$.
\end{enumerate}
\begin{enumerate}[\emph{\textbf{(A3):}}]
\item In both cases, we have that
      $\bar{W}=W-\int_0^\cdot\sqrt{c}\beta_s\ud s$
      is a $\bar{P}$-Brownian motion,
      $\bar{\nu}^L(\dsdx) = Y(s,x)\nu^L(\ud s,\ud x)$ is the
      $\bar{P}$-compensator of $\mu^L$ and $L$ has the
      following canonical decomposition under $\bar{P}$:
\begin{align}\label{can-deco-pt}
 L_t=\bar{b}t +\sqrt{c}\bar{W}_t
    +\int_0^t\int_{\R}x(\mu^L-\bar{\nu}^L)(\ud s,\ud x),
\end{align}
where
\begin{align}\label{drift-pt}
 \bar{b}t = bt +\int_0^tc\beta_s\ud s
               + \int_0^t\int_{\R}x\big(Y(s,x)-1\big)\nu^L(\dsdx).
\end{align}
\end{enumerate}
\end{thm}
\begin{proof}
Theorems III.3.24, III.5.19 and III.5.35 in \citeN{JacodShiryaev03}
yield the result.
\end{proof}

\begin{rem}
In \eqref{girsanov-Y} $\widetilde{\mathcal P}={\mathcal
P}\otimes\mathcal B(\R) $ is the $\sigma$-field of predictable sets
in $\widetilde\Omega=\Omega\times[0,T]\times\R$ and
$M_{\mu^L}^P=\mu^L(\omega; \ud t,\ud x)P(\ud\omega)$ is the positive
measure on $(\Omega\times[0,T]\times\R,\mathcal F\otimes \mathcal
B([0,T])\otimes \mathcal B(\R))$ defined by
\begin{equation}  \label{eq:32}
  M_{\mu^L}^P(W)=E(W*\mu^L)_T,
\end{equation}
for measurable nonnegative functions $W=W(\omega;t,x)$ given on
$\Omega\times[0,T]\times \R$. Now, the \emph{conditional
expectation} $M_{\mu^L}^P\Big(\frac{\textstyle Z}{\textstyle
Z_-}|\widetilde{\mathcal P}\Big)$ is, by definition, the
$M_{\mu^L}^P$-a.s. unique $\widetilde{\mathcal P}$-measurable
function $Y$ with the property
\begin{equation}\label{eq:33}
  M_{\mu^L}^P\Big(\frac{Z}{Z_-} U\Big) = M_{\mu^L}^P (YU),
\end{equation}
for all nonnegative $\widetilde{\mathcal P}$-measurable functions
$U=U(\omega; t,x)$.
\end{rem}

\begin{rem}
Notice that from condition \eqref{Girsanov-moment} and assumption
($\mathbb M$), follows that $L$ has finite first moment under $\P$
as well, i.e.
\begin{equation}
\Et|L_t|<\infty, \;\;\; \text{ for all } \; \ott.
\end{equation}
Verification follows from Proposition \ref{moments} and direct
calculations.
\end{rem}

\begin{rem}
In general, $L$ is \emph{not necessarily} a \lev process under the
measure $\bar{P}$; this depends on the tuple ($\beta,Y$). The
following cases exist.
\begin{description}
\item[(G1)] if ($\beta,Y$) are deterministic and independent
            of time, then $L$ remains a \emph{L\'evy} process under
            $\bar{P}$; its triplet is ($\bar{b},c,Y\cdot\nu$).
\item[(G2)] if ($\beta,Y$) are deterministic but depend on
            time, then $L$ becomes a process with independent
            (but not stationary) increments under $\bar{P}$,
            often called an \emph{additive} process.
\item[(G3)] if ($\beta,Y$) are neither deterministic nor
            independent of time, then we just know that $L$ is
            a \textit{semimartingale} under $\bar{P}$.
\end{description}
\end{rem}

\begin{rem}
Notice that $c$, the diffusion coefficient, and $\mu^L$, the random
measure of jumps of $L$, did not change under the change of measure
from $P$ to $\P$. That happens because $c$ and $\mu^L$ are
\textit{path properties} of the process and do not change under an
equivalent change of measure. Intuitively speaking, \emph{the paths
do not change, the probability of certain paths occurring changes}.
\end{rem}

\begin{example}\label{ex-girsanov}
Assume that $L$ is a \lev process with canonical decomposition
\eqref{can-deco} under $P$. Assume that $P\sim\P$ and the density
process is
\begin{align}\label{RN-density}
Z_t &= \exp\Big[\beta\sqrt{c} W_{t}
        + \int_{0}^{t}\int_{\R}\alpha x(\mu^{L}-\nu^{L})(\dsdx)\\\nonumber
 &\qquad\qquad
        -\Big(\frac{c\beta^2}{2}
        +\int_{\R}(\e^{\alpha x}-1-\alpha x)\nu(\ud x)\Big)t
      \Big],
\end{align}
where $\beta\in\R_{\geqslant0}$ and $\alpha\in\R$ are constants.

Then, comparing \eqref{RN-density} with \eqref{explicit-density}, we
have that the tuple of functions that characterize the change of
measure is ($\beta,Y)=(\beta,f$), where $f(x)=\e^{\alpha x}$.
Because $(\beta,f)$ are deterministic and independent of time, $L$
remains a \lev process under $\P$, its \lev triplet is
($\bar{b},c,f\nu$) and its canonical decomposition is given by
equations \eqref{can-deco-pt} and \eqref{drift-pt}.
\end{example}

Actually, the change of measure of the previous example corresponds
to the so-called \emph{Esscher transformation} or \emph{exponential
tilting}. In chapter 3 of \citeN{Kyprianou06}, one can find a
significantly easier proof of Girsanov's theorem for \lev processes
for the special case of the Esscher transform. Here, we reformulate
the result of example \ref{ex-girsanov} and give a complete proof
(inspired by \citeNP{EberleinPapapantoleon05b}).

\begin{prop}
Let $L=\Lt$ be a \lev process with canonical decomposition
(\ref{can-deco}) under $P$ and assume that $\E[\e^{uL_t}]<\infty$
for all $u\in[-p,p]$, $p>0$. Assume that $P\sim\P$ with density
process \prozess[Z], or conversely, assume that $\P$ is defined via
the Radon-Nikodym derivative $\frac{\ud\P}{\ud P}=Z_T$; here, we
have that
\begin{align}
Z_t
 &= \frac{\e^{\beta L_t^c}\e^{\alpha L_t^d}}{\E[\e^{\beta L_t^c}]\E[\e^{\alpha L_t^d}]}
\end{align}
for $\beta\in\R$ and $|\alpha|<p$. Then, $L$ remains a \lev process
under $\P$, its \lev triplet is ($\bar{b},c,\bar\nu$), where
$\bar\nu=f\cdot\nu$ for $f(x)=\e^{\alpha x}$, and its canonical
decomposition is given by the following equations
\begin{align}
L_t = \bar{b}t + \sqrt{c}\bar{W}_t
     + \int_0^t\int_{\R}x(\mu^L-\bar{\nu}^L)(\ud s,\ud x),
\end{align}
and
\begin{align}
\bar{b} = b + \beta c + \int_{\R}x\big(\e^{\alpha x}-1\big)\nu(\ud x).
\end{align}
\end{prop}
\begin{proof}
Firstly, using Proposition \ref{exp-mg}, we can immediately deduce
that $Z$ is a positive $P$-martingale; moreover, $Z_0=1$. Hence, $Z$
serves as a density process.

Secondly, we will show that $L$ has independent and stationary
increments under $\P$. Using that $L$ has independent and stationary
increments under $P$ and that $Z$ is a $P$-martingale, we arrive at
the following helpful conclusions: for any $B\in\mathcal B(\R)$,
$F_s\in\F_s$ and $0\le s<t\le T$
\begin{enumerate}
 \item $1_{\{L_t-L_s\in B\}}\frac{Z_t}{Z_s}$ is independent of $1_{\{F_s\}}Z_s$ and of $Z_s$;
 \item $E[Z_s]=1$.
\end{enumerate}
Then, we have that
\begin{align*}
\P(\{L_t-L_s\in B\}\cap F_s)
 &= \E\Big[1_{\{L_t-L_s\in B\}}1_{\{F_s\}}Z_t\Big]\\
 &= \E\Big[1_{\{L_t-L_s\in B\}}\frac{Z_t}{Z_s}\Big]\E\big[1_{\{F_s\}}Z_s\big]\\
 &= \E\Big[1_{\{L_t-L_s\in B\}}\frac{Z_t}{Z_s}\Big]\E[Z_s]\E\big[1_{\{F_s\}}Z_s\big]\\
 &= \E\Big[1_{\{L_t-L_s\in B\}}Z_t\Big]\E\big[1_{\{F_s\}}Z_s\big]\\
 &= \P(\{L_t-L_s\in B\})\P(F_s)
\end{align*}
which yields the independence of the increments. Similarly,
regarding the stationarity of the increments of $L$ under $\P$, we
have that
\begin{align*}
\P(\{L_t-L_s\in B\})
 &= \E\big[1_{\{L_t-L_s\in B\}}Z_t\big]\\
 &= \E\big[1_{\{L_t-L_s\in B\}}\frac{Z_t}{Z_s}\big]\E\big[Z_s\big]\\
 &= \E\Big[1_{\{L_t-L_s\in B\}}
    \frac{\e^{\alpha(L_t^c-L_s^c)}\e^{\beta(L_t^d-L_s^d)}}
    {\E[\e^{\alpha(L_t^c-L_s^c)}\e^{\beta(L_t^d-L_s^d)}]}\Big]\\
 &= \E\Big[1_{\{L_t-L_s\in B\}}
    \frac{\e^{\alpha L_{t-s}^c+\beta L_{t-s}^d}}{\E[\e^{\alpha L_{t-s}^c+\beta L_{t-s}^d}]}\Big]\\
 &= \E\big[1_{\{L_{t-s}\in B\}}Z_{t-s}\big]\\
 &= \P(\{L_{t-s}\in B\})
\end{align*}
which yields the stationarity of the increments.

Thirdly, we determine the characteristic function of $L$ under $\P$,
which also yields the triplet and canonical decomposition. Applying
Theorem 25.17 in \citeN{Sato99}, the moment generating function
$M_{L_t}$ of $L_t$ exists for $u\in\mathbb C$ with $\Re u\in[-p,p]$.
We get
\begin{align*}
\bar{\E}\big[\e^{zL_t}\big]
 &= \E\big[\e^{zL_t}Z_t\big]
  = \E\bigg[\frac{\e^{zL_t}\e^{\beta L_t^c}\e^{\alpha L_t^d}}{\E[\e^{\beta L_t^c}]\E[\e^{\alpha L_t^d}]}\bigg]
  = \frac{\E[\e^{zbt}\e^{(z+\beta)L_t^c}\e^{(z+\alpha) L_t^d}]}{\E[\e^{\beta L_t^c}]\E[\e^{\alpha L_t^d}]}\\
 &= \exp\left(t\Big[zb + \frac{(z+\beta)^2c}{2}
   + \int_{\R}\big(\e^{(z+\alpha) x}-1-(z+\alpha)x\big)\nu(\ud x)\right.\\
 &\qquad\qquad\left.- \frac{\beta^2c}{2} - \int_{\R}\big(\e^{\alpha x}-1-\alpha x\big)\nu(\ud x)\Big]\right)\\
 &= \exp\left( t \Big[ z\big(b + \beta c + \int_{\R}x\big(\e^{\alpha x}-1\big)\nu(\ud x)\big)
                       + \frac{z^2c}{2} \right.\\
 &\qquad\qquad \left.  + \int_{\R}\big(\e^{zx}-1-zx\big)\e^{\alpha x}\nu(\ud x) \Big]\right)\\
 &= \exp\left(t\Big[ z\bar b + \frac{z^2c}2 + \int_\R\big(\e^{zx}-1-zx\big)\bar\nu(\dx)\Big]\right).
\end{align*}

Finally, the statement follows by proving that $\bar\nu(\ud
x)=\e^{\alpha x}\nu(\ud x)$ is a \lev measure, i.e.
$\int_{\R}(1\wedge x^2)\e^{\alpha x}\nu(\ud x)<\ap$. It suffices to
note that
\begin{equation}
 \int_{|x|\leq1}x^2\e^{\alpha x}\nu(\ud x)
 \leq C\int_{|x|\leq1}x^2\nu(\ud x) <\ap,
\end{equation}
where $C$ is a positive constant, because $\nu$ is a \lev measure;
the other part follows from the assumptions, since $|\alpha|<p$.
\end{proof}

\begin{rem}
Girsanov's theorem is a very powerful tool, widely used in
mathematical finance. In the second part, it will provide the link
between the `real-world' and the `risk-neutral' measure in a
\lev-driven asset price model. Other applications of Girsanov's
theorem allow to simplify certain valuation problems, cf. e.g.
\citeN{Papapantoleon06} and references therein.
\end{rem}

\section{Construction of L\'evy processes}

Three popular methods to construct a \lev process are described
below.

\begin{description}
 \item[(C1)] Specifying a \emph{\lev triplet}; more specifically, whether there
             exists a Brownian component or not and what is the \lev
             measure. Examples of \lev process constructed this way
             include the standard Brownian motion, which has \lev triplet
             $(0,1,0)$ and the \lev jump-diffusion, which has \lev triplet
             $(b,\sigma^2,\lambda F)$.
 \item[(C2)] Specifying an \emph{infinitely divisible} random variable as the density
             of the increments at time scale 1 (i.e. $L_1$). Examples of \lev
             process constructed this way include the standard Brownian motion,
             where $L_1\sim\text{Normal}(0,1)$ and the normal inverse Gaussian
             process, where $L_1\sim\text{NIG}(\alpha,\beta,\delta,\mu)$.
 \item[(C3)] \emph{Time-changing} Brownian motion with an independent increasing \lev
             process. Let $W$ denote the standard Brownian motion; we can
             construct a \lev process by `replacing' the (calendar) time $t$ by an
             independent increasing \lev process $\tau$, therefore $L_t:=W_{\tau(t)}$, $\ott$.
             The process $\tau$ has the useful -- in Finance -- interpretation as
             `business time'. Models constructed this way include the normal inverse
             Gaussian process, where Brownian motion is time-changed with the inverse
             Gaussian process and the variance gamma process, where Brownian motion
             is time-changed with the gamma process.
\end{description}

Naturally, some processes can be constructed using more than one
methods. Nevertheless, each method has some distinctive advantages
which are very useful in applications. The advantages of specifying
a triplet (C1) are that the characteristic function and the pathwise
properties are known and allows the construction of a rich variety
of models; the drawbacks are that parameter estimation and
simulation (in the infinite activity case) can be quite involved.
The second method (C2) allows the easy estimation and simulation of
the process; on the contrary the structure of the paths might be
unknown. The method of time-changes (C3) allows for easy simulation,
yet estimation might be quite difficult.

\section{Simulation of L\'evy processes}\label{simulation}

We shall briefly describe simulation methods for \lev processes. Our
attention is focused on finite activity \lev processes (i.e. \lev
jump-diffusions) and some special cases of infinite activity \lev
processes, namely the normal inverse Gaussian and the variance gamma
processes. Several speed-up methods for the Monte Carlo simulation
of \lev processes are presented in \citeN{Webber05}.

Here, we do not discuss simulation methods for random variables with
known density; various algorithms can be found in \citeN{Devroye86},
also available online at
\texttt{http://cg.scs.carleton.ca/~luc/rnbookindex.html}.

\subsection{Finite activity}

Assume we want to simulate the \lev jump-diffusion
\begin{align*}
L_{t}= b t + \sigma W_{t}+ \sum_{k=1}^{N_{t}} J_{k}
\end{align*}
where $N_t\sim\text{Poisson}(\lambda t)$ and $J\sim F(\dx)$. $W$
denotes a standard Brownian motion, i.e.
$W_t\sim\text{Normal}(0,t)$.

We can simulate a discretized trajectory of the \lev jump-diffusion
$L$ at fixed time points $t_1,\ldots,t_n$ as follows:
\begin{itemize}
 \item generate a standard normal variate and transform it into a normal
       variate, denoted $G_i$, with variance $\sigma\Delta t_i$,
       where $\Delta t_i=t_i-t_{i-1}$;
 \item generate a Poisson random variate $N$ with parameter $\lambda T$;
 \item generate $N$ random variates $\tau_k$ uniformly distributed in
       $[0,T]$; these variates correspond to the jump times;
 \item simulate the law of jump size $J$, i.e. simulate random variates
       $J_k$ with law $F(\dx)$.
\end{itemize}
The discretized trajectory is
\begin{align*}
L_{t_i}= bt_i + \sum_{j=1}^{i} G_j+ \sum_{k=1}^{N} 1_{\{\tau_k<t_i\}}J_{k}.
\end{align*}

\subsection{Infinite activity}

The variance gamma and the normal inverse Gaussian process can be
easily simulated because they are time-changed Brownian motions; we
follow \citeN{ContTankov03} closely. A general treatment of
simulation methods for infinite activity \lev processes can be found
in \citeN{ContTankov03} and \citeN{Schoutens03}.

Assume we want to simulate a normal inverse Gaussian (NIG) process
with parameters $\alpha,\beta,\delta,\mu$; cf. also section
\ref{NIG-section}. We can simulate a discretized trajectory at fixed
time points $t_1,\ldots,t_n$ as follows:
\begin{itemize}
 \item simulate $n$ independent inverse Gaussian variables $I_i$ with
       parameters $(\delta\Delta t_i)^2$ and $\alpha^2-\beta^2$,
       where $\Delta t_i=t_i-t_{i-1}$, $i=1,\ldots,n$;
 \item simulate $n$ i.i.d. standard normal variables $G_i$;
 \item set $\Delta L_{i} = \mu\Delta t_i + \beta I_i + \sqrt{I_i}G_i$.
\end{itemize}
The discretized trajectory is
\begin{align*}
L_{t_i}= \sum_{k=1}^i \Delta L_{k}.
\end{align*}

Assume we want to simulate a variance gamma (VG) process with
parameters $\sigma,\theta,\kappa$; we can simulate a discretized
trajectory at fixed time points $t_1,\ldots,t_n$ as follows:
\begin{itemize}
 \item simulate $n$ independent gamma variables $\Gamma_i$ with
       parameter $\frac{\Delta t_i}{\kappa}$
 \item set $\Gamma_i=\kappa\Gamma_i$;
 \item simulate $n$ standard normal variables $G_i$;
 \item set $\Delta L_{i} = \theta \Gamma_i + \sigma\sqrt{\Gamma_i}G_i$.
\end{itemize}
The discretized trajectory is
\begin{align*}
L_{t_i}= \sum_{k=1}^i \Delta L_{k}.
\end{align*}

\part{Applications in Finance}

\section{Asset price model}

We describe an asset price model driven by a \lev process, both
under the `real' and under the `risk-neutral' measure. Then, we
present an informal account of market incompleteness.

\subsection{Real-world measure}

Under the real-world measure, we model the asset price process as
the exponential of a \lev process, that is
\begin{align}
S_{t}=S_{0}\exp L_{t}, \quad \ott,
\end{align}
where, $L$ is the \lev process whose infinitely divisible
distribution has been estimated from the data set available for the
particular asset. Hence, the log-returns of the model have
independent and stationary increments, which are distributed --
along time intervals of specific length, e.g. 1 -- according to an
infinitely divisible distribution $\mathcal L(X)$, i.e. $L_1\eqlaw
X$.

Naturally, the path properties of the process $L$ carry over to $S$;
if, for example, $L$ is a pure-jump \lev process, then $S$ is also a
pure-jump process. This fact allows us to capture, up to a certain
extent, the microstructure of price fluctuations, even on an
intraday time scale.

An application of It\^o's formula yields that $S=(S_t)_\ott$ is the
solution of the stochastic differential equation
\begin{align}\label{SDE-L}
\ud S_t=S_{t-}\Big(\ud L_t + \frac{c}{2}\ud t + \int_\R (\e^{x}-1-x)\mu^L(\dt,\dx)\Big).
\end{align}
We could also specify $S$ by replacing the Brownian motion in the
Black--Scholes SDE by a \lev process, i.e. via
\begin{align}\label{SDE-BM}
\ud S_t=S_{t^-}\ud L_t,
\end{align}
whose solution is the stochastic exponential
\begin{align}
S_t=S_0\mathcal{E}(L_t).
\end{align}
The second approach is unfavorable for financial applications,
because ($a$) the asset price can take negative values, unless jumps
are restricted to be larger than $-1$, i.e.
$\text{supp}(\nu)\subset[-1,\infty)$, and ($b$) the distribution of
log-returns is not known. Of course, in the special case of the
Black--Scholes model the two approaches coincide.

\begin{rem}
The two modeling approaches are nevertheless closely related and, in
some sense, complementary of each other. One approach is suitable
for studying the distributional properties of the price process and
the other for investigating the martingale properties. For the
connection between the natural and stochastic exponential for \lev
processes, we refer to Lemma A.8 in \citeN{GollKallsen00}.
\end{rem}

The fact that the price process is driven by a \lev process, makes
the market, in general, incomplete; the only exceptions are the
markets driven by the Normal (Black-Scholes model) and Poisson
distributions. Therefore, there exists a large set of equivalent
martingale measures, i.e. candidate measures for risk-neutral
valuation.

\citeN{EberleinJacod97} provide a thorough analysis and
characterization of the set of equivalent martingale measures for
\lev-driven models. Moreover, they prove that the range of option
prices for a convex payoff function, e.g. a call option, under all
possible equivalent martingale measures spans the whole no-arbitrage
interval, e.g. $[(S_0- K\e^{-rT})^+, S_0]$ for a European call
option with strike $K$. \citeN{Selivanov05} discusses the existence
and uniqueness of martingale measures for exponential \lev models in
finite and infinite time horizon and for various specifications of
the no-arbitrage condition.

The \lev market can be completed using particular assets, such as
moment derivatives (e.g. variance swaps), and then there exists a
\textit{unique} equivalent martingale measure; see
Corcuera, Nualart, and Schoutens
\citeyear{CorcueraNualartSchoutens05,CorcueraNualartSchoutens05b}.
For example, if an asset is driven by a \lev jump-diffusion
\begin{align}
L_{t}= b t + \sqrt{c} W_{t}+ \sum_{k=1}^{N_{t}} J_{k}
\end{align}
where $J_k\equiv\alpha$ $\forall k$, then the market can be
completed using only variance swaps on this asset; this example will
be revisited in section \ref{mark-inc}.

\subsection{Risk-neutral measure}\label{RNM}

Under the risk neutral measure, denoted by $\P$, we model the asset
price process as an exponential \lev process
\begin{align}\label{levyproc}
S_{t}=S_{0}\exp L_{t}
\end{align}
where the \lev process $L$ has the triplet $(\bar b,\bar c,\bar\nu)$
and satisfies Assumptions ($\mathbb{M}$) (cf. Remark \ref{EM}) and
($\mathbb{EM}$) (see below).

The process $L$ has the canonical decomposition
\begin{align}
 L_{t}= \bar bt +\sqrt{\bar c} \bar W_{t}
      + \int_{0}^{t}\int_{\R} x(\mu^{L}-\bar\nu^{L})(\ud s,\ud x)
\end{align}
where $\bar W$ is a $\P$-Brownian motion and $\bar\nu^{L}$ is the
$\P$-compensator of the jump measure $\mu^{L}$.

Because we have assumed that $\P$ is a risk neutral measure,  the
asset price has mean rate of return $\mu\triangleq r-\delta$ and the
discounted and re-invested process $(\e^{(r-\delta)t}S_{t})_\ott$,
is a martingale under $\P$. Here $r\geq0$ is the (domestic)
risk-free interest rate, $\delta\geq0$ the continuous dividend yield
(or foreign interest rate) of the asset. Therefore, the drift term
$\bar b$ takes the form
\begin{align}\label{drift}
\bar b = r -\delta -\frac{\bar c}{2}
   - \int_{\R}(\e^{x}-1-x)\bar\nu(\ud x);
\end{align}
see \citeN{EberleinPapapantoleonShiryaev06} and
\citeN{Papapantoleon06} for all the details.

\begin{assuem}
We assume that the \lev process $L$ has finite first exponential
moment, i.e.
\begin{equation}
\bar\E[\e^{L_t}]<\infty.
\end{equation}
\end{assuem}

There are various ways to choose the martingale measure such that it
is equivalent to the real-world measure.  We refer to
\citeN{GollRueschendorf01} for a unified exposition -- in terms of
$f$-divergences -- of the different methods for selecting an
equivalent martingale measure (EMM). Note that, some of the proposed
methods to choose an EMM preserve the \lev property of log-returns;
examples are the Esscher transformation and the minimal entropy
martingale measure (cf. \citeNP{EscheSchweizer05}).

The market practice is to consider the choice of the martingale
measure as the result of a calibration to market data of vanilla
options. Hakala and Wystup \citeyear{HakalaWystup02b} describe the
calibration procedure in detail. \citeANP{ContTankov04}
\citeyear{ContTankov04,ContTankov05} and \citeN{BelomestnyReiss05}
present numerically stable calibration methods for \lev driven
models.

\subsection{On market incompleteness}\label{mark-inc}

In order to gain a better understanding of why the market is
incomplete, let us make the following observation. Assume that the
price process of a financial asset is modeled as an exponential \lev
process under both the real and the risk-neutral measure. Assume
that these measures, denoted $P$ and $\P$, are equivalent and denote
the triplet of the \lev \proc under $P$ and $\P$ by $(b,c,\nu)$ and
$(\bar b,\bar c,\bar\nu)$ respectively.

Now, applying Girsanov's theorem we get that these triplets are
related via $\bar c=c$, $\bar\nu=Y\cdot\nu$ and
\begin{align}\label{inc-G}
 \bar b = b + c\beta + x(Y-1)*\nu,
\end{align}
where $(\beta,Y)$ is the tuple of functions related to the density
process. On the other hand, from the martingale condition we get
that
\begin{align}\label{inc-M}
 \bar b = r -\frac{\bar c}{2} - (\e^{x}-1-x)*\bar\nu.
\end{align}
Equating \eqref{inc-G} and \eqref{inc-M} and using $c=\bar c$ and
$\nu=Y\cdot\bar\nu$, we have that
\begin{align}\label{inc-eq}
0 &= b + c\beta + x(Y-1)*\nu - r +\frac{\bar c}{2} + (\e^{x}-1-x)*\bar\nu
 \nonumber\\\Leftrightarrow
0 &= b - r + c(\beta+\frac12) + \big((\e^x-1)Y-x\big)*\nu;
\end{align}
therefore, we have \textit{one} equation but \textit{two} unknown
parameters, $\beta$ and $Y$ stemming from the change of measure.
Every \textit{solution} tuple $(\beta,Y)$ of equation \eqref{inc-eq}
corresponds to a different \textit{equivalent martingale measure},
which explains why the market is \textit{not complete}. The tuple
$(\beta,Y)$ could also be termed the tuple of \textit{`market price
of risk'}.

\begin{example}[Black--Scholes model]\label{exam-BM}
Let us consider the Black--Scholes model, where the driving process
is a Brownian motion with drift, i.e. $L_t=bt+\sqrt c W_t$. Then,
equation \eqref{inc-eq} has a \textit{unique} solution, namely
\begin{equation}\label{inc-BM}
 \beta = \frac{r-b}{c} - \frac12,
\end{equation}
the martingale measure is unique and the market is complete. We can
also easily check that plugging \eqref{inc-BM} into \eqref{inc-G},
we recover the martingale condition \eqref{inc-M}.
\end{example}

\begin{rem}
The quantity $\beta$ in \eqref{inc-BM} is nothing else than the
so-called \textit{market price of risk}. The difference from the
quantity often encountered in textbooks, i.e. $\frac{r-\mu}{c}$,
stems from the fact that we model using the \textit{natural} instead
of the \textit{stochastic} exponential, i.e. using SDE \eqref{SDE-L}
and not \eqref{SDE-BM}.
\end{rem}

\begin{example}[Poisson model]\label{exam-P}
Let us consider the Poisson model, where the driving motion is a
Poisson process with intensity $\lambda>0$ and jump size $\alpha$,
i.e. $L_t=bt+\alpha N_t$ and $\nu(\dx)=\lambda 1_{\{\alpha\}}(\dx)$.
Then, equation \eqref{inc-eq} has a \textit{unique} solution for
$Y$, which is
\begin{align}\label{inc-P}
0 &= b - r + \big((\e^x-1)Y-x\big)*\lambda 1_{\{\alpha\}}(\dx)
 \nonumber\\\Leftrightarrow
0 &= b - r + \big((\e^\alpha-1)Y-\alpha\big)\lambda
 \nonumber\\\Leftrightarrow
Y &= \frac{r-b+\alpha\lambda}{(\e^\alpha-1)\lambda};
\end{align}
therefore the martingale measure is unique and the market is
complete. By the analogy to the Black--Scholes case, we could call
the quantity $Y$ in \eqref{inc-P} the \textit{market price of jump
risk}.

Moreover, we can also check that plugging \eqref{inc-P} into
\eqref{inc-G}, we recover the martingale condition \eqref{inc-M};
indeed, we have that
\begin{align*}
\bar b
 &= b+ \alpha\lambda(Y-1)\\
 &= b+ \alpha\lambda(Y-1) + (\e^\alpha-1)Y\lambda - (\e^\alpha-1)Y\lambda\\
 &= r - (\e^\alpha-1-\alpha)\bar\lambda,
\end{align*}
where we have used \eqref{inc-P} and that $\bar\nu=Y\cdot\nu$, which
in the current framework translates to $\bar\lambda=Y\lambda$.
\end{example}

\begin{example}[A simple incomplete model]
Assume that the driving process consists of a drift, a Brownian
motion and a Poisson process, i.e. $L_t=bt+\sqrt c W_t+\alpha N_t$,
as in examples \ref{exam-BM} and \ref{exam-P}. Based on
\eqref{inc-BM} and \eqref{inc-P} we postulate that the solutions of
equation \eqref{inc-eq} are of the form
\begin{align}
\beta_\varepsilon = \varepsilon\frac{r-b}{c} - \frac12
\quad\text{ and }\quad
Y_\varepsilon = \frac{(1-\varepsilon)(r-b)+\alpha\lambda}{(\e^\alpha-1)\lambda}
\end{align}
for any $\varepsilon\in(0,1)$. One can easily verify that
$\beta_\varepsilon$ and $Y_\varepsilon$ satisfy \eqref{inc-eq}. But
then, to \textit{any} $\varepsilon\in(0,1)$ corresponds an
\textit{equivalent martingale measure} and we can easily conclude
that this simple market is \textit{incomplete}.
\end{example}

\section{Popular models}

In this section, we review some popular models in the mathematical
finance literature from the point of view of \lev processes. We
describe their \lev triplets and characteristic functions and
provide, whenever possible, their -- infinitely divisible -- laws.

\subsection{Black--Scholes}

The most famous asset price model based on a \lev process is that of
\citeN{Samuelson65}, \citeN{BlackScholes73} and Merton \citeyear{Merton73}.
The log-returns are normally distributed with mean $\mu$ and
variance $\sigma^2$, i.e. $L_1\sim\text{Normal} (\mu, \sigma^2)$ and
the density is
\begin{align*}
f_{L_1}(x) = \frac{1}{\sigma\sqrt{2\pi}}
         \exp\Big[-\frac{(x-\mu)^2}{2\sigma^2}\Big].
\end{align*}
The characteristic function is
\begin{align*}
\varphi_{L_1}(u) = \exp\Big[i\mu u -\frac{\sigma^2u^2}{2}\Big],
\end{align*}
the first and second moments are
\begin{align*}
\text{E}[L_1]=\mu, &\qquad \text{Var}[L_1]=\sigma^2,
\end{align*}
while the skewness and kurtosis are
\begin{align*}
\text{skew}[L_1]=0, &\qquad \text{kurt}[L_1]=3.
\end{align*}
The canonical decomposition of $L$ is
\begin{align*}
L_t = \mu t + \sigma W_t
\end{align*}
and the \lev triplet is $(\mu, \sigma^2, 0)$.

\subsection{Merton}

\citeN{Merton76} was one of the first to use a discontinuous price
process to model asset returns. The canonical decomposition of the
driving process is
\begin{align*}
L_t = \mu t + \sigma W_t + \sum_{k=1}^{N_t}J_k
\end{align*}
where $J_k\sim\text{Normal} (\mu_J, \sigma^2_J)$, $k=1,...$, hence
the distribution of the jump size has density
\begin{align*}
f_J(x) = \frac{1}{\sigma_J\sqrt{2\pi}}
         \exp\Big[-\frac{(x-\mu_J)^2}{2\sigma^2_J}\Big].
\end{align*}
The characteristic function of $L_1$ is
\begin{align*}
\varphi_{L_1}(u) = \exp\Big[i\mu u - \frac{\sigma^2u^2}{2}
                 + \lambda \big(\e^{i\mu_Ju-\sigma^2_Ju^2/2}-1\big)
                   \Big],
\end{align*}
and the \lev triplet is $(\mu, \sigma^2, \lambda\times f_J)$.

The density of $L_1$ is not known in closed form, while the first
two moments are
\begin{align*}
 \text{E}[L_1]=\mu+\lambda\mu_J
 &\qquad \text{and} \qquad
 \text{Var}[L_1]=\sigma^2+\lambda\mu_J^2+\lambda\sigma^2_J
\end{align*}

\subsection{Kou}

\citeN{Kou02} proposed a jump-diffusion model similar to Merton's,
where the jump size is double-exponentially distributed. Therefore,
the canonical decomposition of the driving process is
\begin{align*}
L_t = \mu t + \sigma W_t + \sum_{k=1}^{N_t}J_k
\end{align*}
where $J_k\sim\text{DbExpo} (p,\theta_1,\theta_2)$, $k=1,...$, hence
the distribution of the jump size has density
\begin{align*}
f_J(x) = p\theta_1\e^{-\theta_1x}\1_{\{x<0\}}
       + (1-p)\theta_2\e^{\theta_2x}\1_{\{x>0\}}.
\end{align*}
The characteristic function of $L_1$ is
\begin{align*}
\varphi_{L_1}(u) = \exp\Big[i\mu u - \frac{\sigma^2u^2}{2}
                 + \lambda \Big(\frac{p\theta_1}{\theta_1-iu}-\frac{(1-p)\theta_2}{\theta_2+iu}-1\Big)
                   \Big],
\end{align*}
and the \lev triplet is $(\mu, \sigma^2, \lambda\times f_J)$.

The density of $L_1$ is not known in closed form, while the first
two moments are
\begin{align*}
 \text{E}[L_1]=\mu+\frac{\lambda p}{\theta_1}-\frac{\lambda(1-p)}{\theta_2}
 &\qquad \text{and} \qquad
 \text{Var}[L_1]=\sigma^2+\frac{\lambda p}{\theta_1^2}+\frac{\lambda(1-p)}{\theta_2^2}.
\end{align*}

\subsection{Generalized Hyperbolic}\label{gh-section}

The generalized hyperbolic model was introduced by
\citeN{EberleinPrause02} following the seminal work on the
hyperbolic model by \citeN{EberleinKeller95}. The class of
hyperbolic distributions was invented by O. E. Barndorff-Nielsen in
relation to the so-called `sand project' (cf.
\citeNP{Barndorff-Nielsen77}). The increments of time length 1
follow a generalized hyperbolic distribution with parameters
$\alpha,\beta,\delta,\mu,\lambda$, i.e.
$L_1\sim\text{GH}(\alpha,\beta,\delta,\mu,\lambda)$ and the density
is
\begin{align*}
 f_{GH}(x) &= c(\lambda,\alpha,\beta,\delta)
               \big(\delta^2+(x-\mu)^2\big)^{(\lambda-\frac{1}{2})/2}\\
      &\qquad
               \times K_{\lambda-\frac{1}{2}}
               \big(\alpha\sqrt{\delta^2+(x-\mu)^2}\big)
               \exp\big(\beta(x-\mu)\big),
\end{align*}
where
\begin{align*}
  c(\lambda,\alpha,\beta,\delta) =
    \frac{(\alpha^2-\beta^2)^{\lambda/2}}
    {\sqrt{2\pi}\alpha^{\lambda-\frac{1}{2}}K_{\lambda}\big(\delta\sqrt{\alpha^2-\beta^2}\big)}
\end{align*}
and $K_{\lambda}$ denotes the Bessel function of the third kind with
index $\lambda$ (cf. \citeNP{AbramowitzStegun68}). Parameter
$\alpha>0$ determines the shape, $0\leq|\beta|<\alpha$ determines
the skewness, $\mu\in\R$ the location and $\delta>0$ is a scaling
parameter. The last parameter, $\lambda\in\R$ affects the heaviness
of the tails and allows us to navigate through different subclasses.
For example, for $\lambda=1$ we get the hyperbolic distribution and
for $\lambda=-\frac{1}{2}$ we get the normal inverse Gaussian (NIG).

The characteristic function of the GH distribution is
\begin{align*}
 \varphi_{GH}(u) = \e^{iu\mu}
   \Big(\frac{\alpha^2-\beta^2}{\alpha^2-(\beta+iu)^2}\Big)^{\frac{\lambda}{2}}
   \frac{K_{\lambda}\big(\delta\sqrt{\alpha^2-(\beta+iu)^2}\big)}
   {K_{\lambda}\big(\delta\sqrt{\alpha^2-\beta^2}\big)},
\end{align*}
while the first and second moments are
\begin{align*}
 \text{E}[L_1]= \mu +
      \frac{\beta\delta^2}{\zeta}\frac{K_{\lambda+1}(\zeta)}{K_{\lambda}(\zeta)}
\end{align*}
and
\begin{align*}
 \text{Var}[L_1]=\frac{\delta^2}{\zeta}\frac{K_{\lambda+1}(\zeta)}{K_{\lambda}(\zeta)}
     + \frac{\beta^2\delta^4}{\zeta^2}
       \Big(\frac{K_{\lambda+2}(\zeta)}{K_{\lambda}(\zeta)}-
            \frac{K^2_{\lambda+1}(\zeta)}{K^2_{\lambda}(\zeta)}\Big),
\end{align*}
where $\zeta=\delta\sqrt{\alpha^2-\beta^2}$.

The canonical decomposition of a \lev process driven by a
generalized hyperbolic distribution (i.e. $L_1\sim\text{GH}$) is
\begin{align*}
 L_t=t\text{E}[L_1] + \int_0^t\int_{\R}x(\mu^L-\nu^{GH})(\dsdx)
\end{align*}
and the \lev triplet is ($E[L_1],0,\nu^{GH}$). The \lev measure of
the GH distribution has the following form
\begin{displaymath}
\nu^{GH}(\dx) = \frac{\e^{\beta x}}{|x|}
   \left(\int_0^\infty \frac{\exp(-\sqrt{2y+\alpha^2}\,|x|)}
                            {\pi^2y(J_{|\lambda|}^2 (\delta\sqrt{2y}\kern1pt)
              +Y_{|\lambda|}^2(\delta\sqrt{2y}\kern1pt))}\ud y
            +\lambda \e^{-\alpha|x|} 1_{\{\lambda\geq0\}}\right);
\end{displaymath}
here $J_{\lambda}$ and $Y_{\lambda}$ denote the Bessel functions of
the first and second kind with index $\lambda$. We refer to
\citeN[section 2.4.1]{Raible00} for a fine analysis of this \lev
measure.

The GH distribution contains as special or limiting cases several
known distributions, including the normal, exponential, gamma,
variance gamma, hyperbolic and normal inverse Gaussian
distributions; we refer to Eberlein and v. Hammerstein
\citeyear{EberleinHammerstein04} for an exhaustive survey.

\subsection{Normal Inverse Gaussian}\label{NIG-section}

The normal inverse Gaussian distribution is a special case of the GH
for $\lambda=-\frac{1}{2}$; it was introduced to finance in
\citeN{Barndorff-Nielsen97}. The density is
\begin{align*}
 f_{NIG}(x) &= \frac{\alpha}{\pi}
   \exp\big(\delta\sqrt{\alpha^2-\beta^2}+\beta(x-\mu)\big)
   \frac{K_{1}\Big(\alpha\delta\sqrt{1+(\frac{x-\mu}{\delta})^2}\Big)}{\sqrt{1+(\frac{x-\mu}{\delta})^2}},
\end{align*}
while the characteristic function has the simplified form
\begin{align*}
 \varphi_{NIG}(u) = \e^{iu\mu}
   \frac{\exp(\delta\sqrt{\alpha^2-\beta^2})}{\exp(\delta\sqrt{\alpha^2-(\beta+iu)^2})}.
\end{align*}
The first and second moments of the NIG distribution are
\begin{align*}
 \text{E}[L_1]=\mu+\frac{\beta\delta}{\sqrt{\alpha^2-\beta^2}}
 &\qquad \text{and} \qquad
 \text{Var}[L_1]=\frac{\delta}{\sqrt{\alpha^2-\beta^2}}+\frac{\beta^2\delta}{(\sqrt{\alpha^2-\beta^2})^3},
\end{align*}
and similarly to the GH, the canonical decomposition is
\begin{align*}
 L_t=t\text{E}[L_1] + \int_0^t\int_{\R}x(\mu^L-\nu^{NIG})(\dsdx),
\end{align*}
where now the \lev measure has the simplified form
\begin{align*}
 \nu^{NIG}(\dx) = \e^{\beta x} \frac{\delta\alpha}{\pi|x|} K_1(\alpha|x|)\dx.
\end{align*}

The NIG is the only subclass of the GH that is closed under
convolution, i.e. if $X\sim\text{NIG}(\alpha,\beta,\delta_1,\mu_1)$
and $Y\sim\text{NIG}(\alpha,\beta,\delta_2,\mu_2)$ and $X$ is
independent of $Y$, then
\[
X+Y\sim\text{NIG}(\alpha,\beta,\delta_1+\delta_2,\mu_1+\mu_2).
\]
Therefore, if we estimate the returns distribution at some time
scale, then we know it -- in closed form -- for all time scales.

\subsection{CGMY}

The CGMY \lev process was introduced by Carr, Geman, Madan, and Yor
\citeyear{Carretal02}; another
name for this process is (generalized) tempered stable process (see
e.g. \citeNP{ContTankov03}). The characteristic function of $L_t$,
$t\in[0,T]$ is
\begin{align*}
 \varphi_{L_t}(u)
 = \exp\Big( tC\Gamma(-Y)\big[(M-iu)^Y+(G+iu)^Y-M^Y-G^Y\big] \Big).
\end{align*}
The \lev measure of this process admits the representation
\begin{align*}
 \nu^{CGMY}(\ud x) = C\frac{\e^{-Mx}}{x^{1+Y}}\1_{\{x>0\}}\ud x
            + C\frac{\e^{Gx}}{|x|^{1+Y}}\1_{\{x<0\}}\ud x,
\end{align*}
where $C>0$, $G>0$, $M>0$,  and $Y<2$. The CGMY process is a pure
jump \lev process with canonical decomposition
\begin{align*}
 L_t=t\text{E}[L_1] + \int_0^t\int_{\R}x(\mu^L-\nu^{CGMY})(\dsdx),
\end{align*}
and \lev triplet ($\text{E}[L_1],0,\nu^{CGMY}$), while the density is not
known in closed form.

The CGMY processes are closely related to stable processes; in fact,
the \lev measure of the CGMY process coincides with the \lev measure
of the stable process with index $\alpha\in(0,2)$ (cf. \citeNP[Def.
1.1.6]{SamorodnitskyTaqqu94}), \textit{but} with the additional
exponential factors; hence the name \textit{tempered} stable
processes. Due to the exponential tempering of the \lev measure, the
CGMY distribution has finite moments of all orders. Again, the class
of CGMY distributions contains several other distributions as
subclasses, for example the variance gamma distribution
(Madan and Seneta \citeyearNP{MadanSeneta90}) and the bilateral gamma
distribution (K\"uchler and Tappe \citeyearNP{KuechlerTappe06}).

\subsection{Meixner}

The Meixner process was introduced by Schoutens and Teugels \citeyear{SchoutensTeugels98},
see also \citeN{Schoutens02}. Let $L=(L_t)_{0\le t\le T}$ be a
Meixner process with
$\mathrm{Law}(H_1|P)=\mathrm{Meixner}(\alpha,\beta,\delta)$,
$\alpha>0$, $-\pi<\beta<\pi$, $\delta>0$, then the density is
\begin{align*}
 f_{\mathrm{Meixner}}(x)
  = \frac{\left(2\cos\frac\beta2\right)^{2\delta}}{2\alpha\pi\Gamma(2\delta)}
    \exp \left( \frac{\beta x}{\alpha} \right)
    \left| \Gamma\left(\delta+\frac{ix}{\alpha}\right) \right|^2.
\end{align*}
The characteristic function $L_t$, $t\in[0,T]$ is
\begin{align*}
 \varphi_{L_t}(u)
  = \left(\frac{\cos\frac\beta2}{\cosh\frac{\alpha u-i\beta}{2}}\right)^{2\delta t},
\end{align*}
and the \lev measure of the Meixner process admits the
representation
\begin{displaymath}
 \nu^{\mathrm{Meixner}}(\ud x)
   =\frac{\delta \exp\left(\frac{\beta}{\alpha}x\right)}{x\sinh(\frac{\pi x}{\alpha})}.
\end{displaymath}
The Meixner process is a pure jump \lev process with canonical
decomposition
\begin{align*}
 L_t=t\text{E}[L_1] + \int_0^t\int_{\R}x(\mu^L-\nu^{\mathrm{Meixner}})(\dsdx),
\end{align*}
and \lev triplet ($\text{E}[L_1],0,\nu^{\mathrm{Meixner}}$).

\section{Pricing European options}

The aim of this section is to review the three predominant methods
for pricing European options on assets driven by general \lev
processes. Namely, we review transform methods, partial
integro-differential equation (PIDE) methods and Monte Carlo
methods. Of course, all these methods can be used -- under certain
modifications -- when considering more general driving processes as
well.

The setting is as follows: we consider an asset \prozess[S] modeled
as an exponential \lev process, i.e.
\begin{align}
S_{t}=S_{0}\exp L_{t}, \quad \ott,
\end{align}
where \prozess has the \lev triplet $(b,c,\nu)$. We assume that the
asset is modeled directly under a martingale measure, cf. section
\ref{RNM}, hence the martingale restriction on the drift term $b$ is
in force. For simplicity, we assume that $r>0$ and $\delta=0$
throughout this section.

We aim to derive the price of a European option on the asset $S$
with payoff function $g$ maturing at time $T$, i.e. the payoff of
the option is $g(S_T)$.

\subsection{Transform methods}

The simpler, faster and most common method for pricing European
options on assets driven by \lev processes is to derive an integral
representation for the option price using Fourier or Laplace
transforms. This blends perfectly with \lev processes, since the
representation involves the characteristic function of the random
variables, which is explicitly provided by the \lev-Khintchine
formula. The resulting integral can be computed numerically very
easily and fast. The main drawback of this method is that exotic
derivatives cannot be handled so easily.

Several authors have derived valuation formulae using Fourier or
Laplace transforms, see e.g. \citeN{CarrMadan99},
\citeN{BorovkovNovikov02} and \citeN{EberleinGlauPapapantoleon08}. Here,
we review the method developed by S. Raible (cf. \citeNP[Chapter
3]{Raible00}).

Assume that the following conditions regarding the driving process
of the asset and the payoff function are in force.
\begin{description}
\item[(T1)] Assume that $\varphi_{L_T}(z)$,
 the extended characteristic function of $L_T$, exists
 for all $z\in\C$ with $\Im z\in I_1\supset[0,1]$.
\item[(T2)] Assume that $P_{L_T}$, the distribution
 of $L_T$, is absolutely continuous w.r.t. the Lebesgue measure
 $\llambda$ with density $\rho$.
\item[(T3)] Consider an integrable, European-style, payoff function $g(S_T)$.
\item[(T4)] Assume that $x\mapsto \e^{-Rx}|g(\e^{-x})|$ is bounded and
 integrable for all $R\in I_2\subset\R$.
\item[(T5)] Assume that $I_1\cap I_2\neq\emptyset$.
\end{description}

Furthermore, let $\mathfrak{L}_h(z)$ denote the bilateral Laplace
transform of a function $h$ at $z\in\C$, i.e. let
\begin{align*}
\mathfrak{L}_h(z):=\int_{\R} \e^{-zx}h(x)\ud  x.
\end{align*}

According to arbitrage pricing, the value of an option is equal to
its discounted expected payoff under the risk-neutral measure $P$.
Hence, we get
\begin{align*}
C_T(S,K)
 &= \e^{-rT}\E [g(S_T)]
  = \e^{-rT}\int_{\Omega} g(S_T) \ud P\\
 &= \e^{-rT}\int_{\R} g(S_0\e^x) \ud P_{L_T}(x)
  = \e^{-rT}\int_{\R} g(S_0\e^x) \rho(x)\ud x
\end{align*}
because $P_{L_T}$ is absolutely continuous with respect to the
Lebesgue measure. Define the function $\pi(x)=g(\e^{-x})$ and let
$\zeta=-\log S_0$, then
\begin{align}\label{conv-rep}
C_T(S,K)
 = \e^{-rT}\int_{\R}\pi(\zeta-x)\rho(x)\ud x
 = \e^{-rT}(\pi\ast\rho)(\zeta)
 =: C
\end{align}
which is a convolution of $\pi$ with $\rho$ at the point $\zeta$,
multiplied by the discount factor.

The idea now is to apply a Laplace transform on both sides of
\eqref{conv-rep} and take advantage of the fact that \emph{the
Laplace transform of a convolution equals the product of the Laplace
transforms of the factors}. The resulting Laplace transforms are
easier to calculate analytically. Finally, we can invert the Laplace
transforms to recover the option value.

Applying Laplace transforms on both sides of \eqref{conv-rep} for
$\C\ni z=R+iu, R\in I_1\cap I_2, u\in\R$, we get that
\begin{align*}
\mathfrak{L}_C(z)
 &= \e^{-rT}\int_{\R}\e^{-zx}(\pi\ast\rho)(x)\ud x\\
 &= \e^{-rT}\int_{\R}\e^{-zx}\pi(x)\ud x \int_{\R}\e^{-zx}\rho(x)\ud x\\
 &= \e^{-rT}\mathfrak{L}_\pi(z)\mathfrak{L}_{\rho}(z).
\end{align*}
Now, inverting this Laplace transform yields the option value, i.e.
\begin{align*}
C_T(S,K)
 &= \frac{1}{2\pi i} \int_{R-i\infty}^{R+i\infty}
    \e^{\zeta z}\mathfrak{L}_C(z)\ud z\\
 &= \frac{1}{2\pi} \int_{\R}
    \e^{\zeta (R+iu)} \mathfrak{L}_C(R+iu)\ud u\\
 &= \frac{\e^{\zeta R}}{2\pi} \int_{\R}
    \e^{i\zeta u}\e^{-rT} \mathfrak{L}_\pi(R+iu) \mathfrak{L}_{\rho}(R+iu)\ud u\\
 &= \frac{\e^{-rT+\zeta R}}{2\pi} \int_{\R}
    \e^{i\zeta u}\mathfrak{L}_{\pi}(R+iu) \varphi_{L_T}(iR-u)\ud u.
\end{align*}
Here, $\mathfrak{L}_{\pi}$ is the Laplace transform of the modified
payoff function $\pi(x)=g(\e^{-x})$ and $\varphi_{L_T}$ is provided
directly from the \lev-Khintchine formula. Below, we describe two
important examples of payoff functions and their Laplace transforms.

\begin{example}[Call and put option]
A European call option pays off $g(S_T)=(S_T-K)^+$, for some strike
price $K$. The Laplace transform of its modified payoff function
$\pi$ is
\begin{align}\label{lt-eucall}
\mathfrak{L}_\pi(z)=\frac{K^{1+z}}{z(z+1)}
\end{align}
for $z\in\C$ with $\Re z=R\in I_2=(-\infty, -1)$.\par Similarly, for
a European put option that pays off $g(S_T)=(K-S_T)^+$, the Laplace
transform of its modified payoff function $\pi$ is given by
\eqref{lt-eucall} for $z\in\C$ with $\Re z=R\in I_2=(0, \infty)$.
\end{example}

\begin{example}[Digital option]
A European digital call option pays off $g(S_T)=\1_{\{S_T>K\}}$. The
Laplace transform of its modified payoff function $\pi$ is
\begin{align}
\mathfrak{L}_\pi(z)=-\frac{K^z}{z}
\end{align}
for $z\in\C$ with $\Re z=R\in I_2=(-\infty, 0)$.

Similarly, for a European digital put option that pays off
$g(S_T)=\1_{\{S_T<K\}}$, the Laplace transform of its modified
payoff function $\pi$ is
\begin{align}
\mathfrak{L}_\pi(z)=\frac{K^z}{z}
\end{align}
for $z\in\C$ with  $\Re z=R\in I_2=(0,\infty)$.
\end{example}

\subsection{PIDE methods}

An alternative to transform methods for pricing options is to derive
and then solve numerically the partial integro-differential equation
(PIDE) that the option price satisfies. Note that in their seminal
paper Black and Scholes derive such a PDE for the price of a
European option. The advantage of PIDE methods is that complex and
exotic payoffs can be treated easily; the limitations are the slower
speed in comparison to transform methods and the computational
complexity when handling options on several assets.

Here, we derive the PIDE corresponding to the price of a European
option in a \lev-driven asset, using martingale techniques; of
course, we could derive the same PIDE by constructing a
self-financing portfolio.

Let us denote by $G(S_t,t)$ the time-$t$ price of a European option
with payoff function $g$ on the asset $S$; the price is given by
\begin{align}
G(S_t,t) = \e^{-r(T-t)}\E[g(S_T)] =: V_t, \quad \ott.
\end{align}
By arbitrage theory, we know that the discounted option price
process must be a martingale under a martingale measure. Therefore,
any decomposition of the price process as
\begin{align}
\e^{-rt}V_t = V_0 + M_t + A_t,
\end{align}
where $M\in\mathcal M_\text{loc}$ and $A\in\mathcal A_\text{loc}$,
must satisfy $A_t=0$ for all $t\in[0,T]$. This condition yields the
desired PIDE.

Now, for notational but also computational convenience, we work with
the driving process $L$ and not the asset price process $S$, hence
we derive a PIDE involving $f(L_t,t):=G(S_t,t)$, or in other words
\begin{align}
f(L_t,t) = \e^{-r(T-t)}\E[g(S_0\e^{L_T})] = V_t, \quad \ott.
\end{align}
Let us denote by $\partial_if$ the derivative of $f$ with respect to
the $i$-th argument, $\partial_i^2f$ the second derivative of $f$
with respect to the $i$-th argument, and so on.

Assume that $f\in C^{2,1}(\R\times[0,T])$, i.e. it is twice
continuously differentiable in the first argument and once
continuously differentiable in the second argument. An application
of It\^o's formula yields:
\begin{align*}
\ud (\e^{-rt}V_t)
 &= \ud (\e^{-rt}f(L_{t-},t)) \\
 &= -r\e^{-rt}f(L_{t-},t)\dt + \e^{-rt}\partial_2f(L_{t-},t)\dt\\
 &\quad + \e^{-rt}\partial_1f(L_{t-},t)\ud L_t
  + \frac12 \e^{-rt}\partial^2_1f(L_{t-},t)\ud\langle L_t^c\rangle\\
 &\quad + \e^{-rt}\int_\R \Big( f(L_{t_-}+z,t)-f(L_{t-},t)-\partial_1f(L_{t_-},t)z \Big)\mu^L(\ud z,\dt)\\
 &= \e^{-rt}\bigg\{ -rf(L_{t-},t)\dt + \partial_2f(L_{t-},t)\dt + \partial_1f(L_{t-},t)b\dt\\
 &\quad + \partial_1f(L_{t-},t)\sqrt{c}\ud W_t + \int_\R\partial_1f(L_{t-},t)z (\mu^L-\nu^L)(\ud z,\dt)\\
 &\quad + \frac12 \partial^2_1f(L_{t-},t)c\dt\\
 &\quad + \int_\R\!\!\Big( f(L_{t_-}+z,t)-f(L_{t-},t)-\partial_1f(L_{t_-},t)z \Big)(\mu^L-\nu^L)(\ud z,\dt)\\
 &\quad + \int_\R\!\!\Big( f(L_{t_-}+z,t)-f(L_{t-},t)-\partial_1f(L_{t_-},t)z \Big)\nu(\ud z)\dt\bigg\}.
\end{align*}
Now, the stochastic differential of the bounded variation part of
the option price process is
\begin{align*}
\e^{-rt}\bigg\{ -rf(L_{t-},t) + \partial_2f(L_{t-},t)
                + \partial_1f(L_{t-},t)b + \frac12 \partial^2_1f(L_{t-},t)c\\
 \qquad\qquad\quad + \int_\R\!\Big( f(L_{t_-}+z,t)-f(L_{t-},t)
                   -\partial_1f(L_{t_-},t)z \Big)\nu(\ud z) \bigg\},
\end{align*}
while the remaining parts constitute of the local martingale part.

As was already mentioned, the bounded variation part vanishes
identically. Hence, the price of the option satisfies the partial
integro-differential equation
\begin{align}
0 &= -rf(x,t) + \partial_2f(x,t) + \partial_1f(x,t)b + \frac{c}2 \partial^2_1f(x,t)\\
  &\quad + \int_\R\!\Big( f(x+z,t)-f(x,t)-\partial_1f(x,t)z \Big)\nu(\ud z), \nonumber
\end{align}
for all $(x,t)\in\R\times(0,T)$, subject to the terminal condition
\begin{align}
f(x,T)=g(\e^x).
\end{align}

\begin{rem}
Using the martingale condition \eqref{drift} to make the drift term
explicit, we derive an equivalent formulation of the PIDE:
\begin{align*}
0 &= -rf(x,t) + \partial_2f(x,t)
   + \big(r-\frac{c}2\big)\partial_1f(x,t) + \frac{c}2 \partial^2_1f(x,t)\\
  &\quad + \int_\R\!\Big( f(x+z,t)-f(x,t)-(\e^z-1)\partial_1f(x,t) \Big)\nu(\ud z),
\end{align*}
for all $(x,t)\in\R\times(0,T)$, subject to the terminal condition
\begin{align*}
f(x,T)=g(\e^x).
\end{align*}
\end{rem}

\begin{rem}
Numerical methods for solving the above partial integro-differential
equations can be found, for example, in
\shortciteN{MatachePetersdorffSchwab04}, in
\shortciteN{MatacheSchwabWihler05} and in \citeN[Chapter
12]{ContTankov03}.
\end{rem}

\subsection{Monte Carlo methods}

Another method for pricing options is to use a Monte Carlo
simulation. The main advantage of this method is that complex and
exotic derivatives can be treated easily -- which is very important
in applications, since little is known about functionals of \lev
processes. Moreover, options on several assets can also be handled
easily using Monte Carlo simulations. The main drawback of Monte
Carlo methods is the slow computational speed.

We briefly sketch the pricing of a European call option on a \lev
driven asset. The payoff of the call option with strike $K$ at the
time of maturity $T$ is $g(S_T)=(S_T-K)^+$ and the price is provided
by the discounted expected payoff under a risk-neutral measure, i.e.
\begin{align*}
C_T(S,K) = \e^{-rT}\E [(S_T-K)^+].
\end{align*}
The crux of pricing European options with Monte Carlo methods is to
simulate the terminal value of asset price $S_T=S_0\exp L_T$ -- see
section \ref{simulation} for simulation methods for \lev processes.
Let $S_{T_k}$ for $k=1,\dots,N$ denote the simulated values; then,
the option price $C_T(S,K)$ is estimated by the average of the
prices for the simulated asset values, that is
\begin{align*}
\widehat C_T(S,K) = \e^{-rT}\sum_{k=1}^N (S_{T_k}-K)^+,
\end{align*}
and by the Law of Large Numbers we have that
\begin{align*}
\widehat C_T(S,K) \rightarrow C_T(S,K)
 \quad \text{as} \quad N\rightarrow\infty.
\end{align*}

\section{Empirical evidence}

\lev processes provide a framework that can easily capture the
empirical observations both under the ``real world'' and under the
``risk-neutral'' measure. We provide here some indicative examples.

Under the ``real world'' measure, \lev processes are generated by
distributions that are flexible enough to capture the observed
fat-tailed and skewed (leptokurtic) behavior of asset returns. One
such class of distributions is the class of \emph{generalized
hyperbolic} distributions (cf. section \ref{gh-section}). In Figure
\ref{GH-examples}, various densities of generalized hyperbolic
distributions and a comparison of the generalized hyperbolic and
normal density are plotted.

\begin{figure}
\begin{center}
  \includegraphics[width=6.cm,keepaspectratio=true]{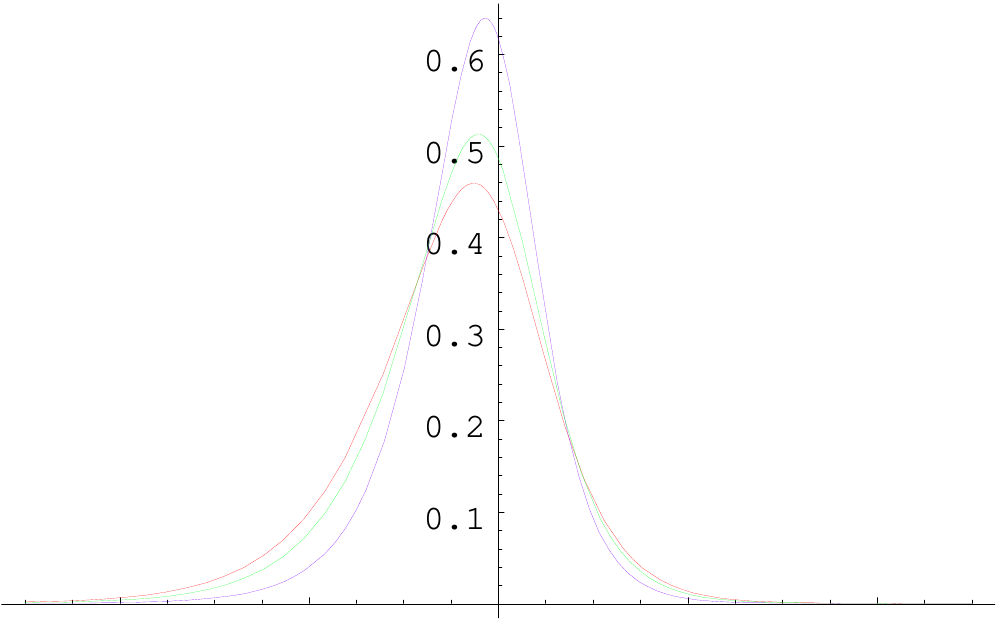}
  \includegraphics[width=6.cm,keepaspectratio=true]{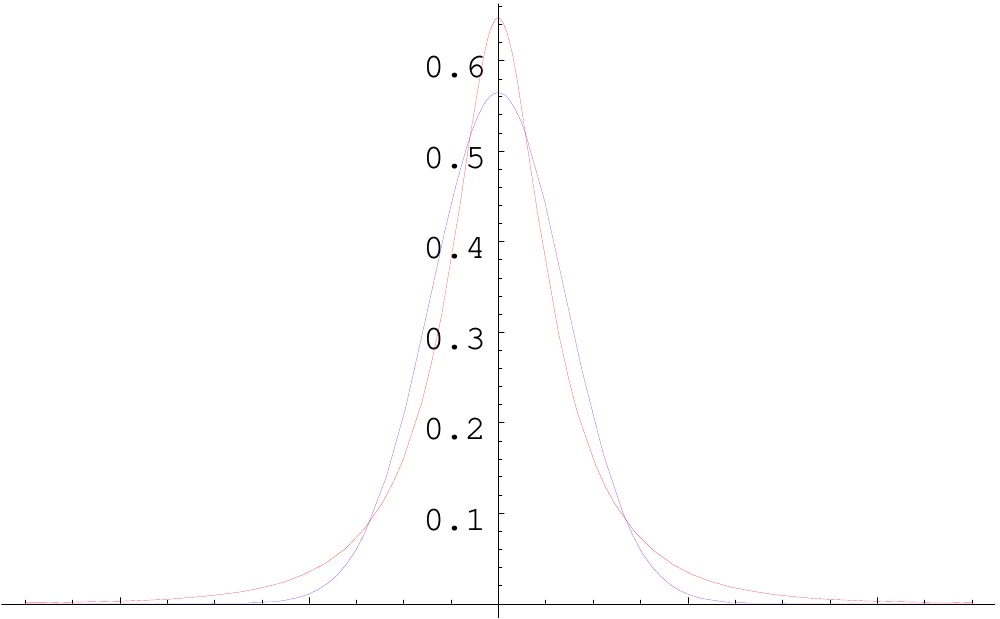}
  \caption{Densities of hyperbolic (red), NIG (blue) and hyperboloid 
           distributions (left). Comparison of the GH (red) and Normal 
           distributions (with equal mean and variance).}
   \label{GH-examples}
\end{center}
\end{figure}

A typical example of the behavior of asset returns can be seen in
Figures \ref{returns} and \ref{GH-fit}. The fitted normal
distribution has lower peak, fatter flanks and lighter tails than
the empirical distribution; this means that, in reality, tiny and
large price movements occur \textit{more frequently}, and small and
medium size movements occur \textit{less frequently}, than predicted
by the normal distribution. On the other hand, the generalized
hyperbolic distribution gives a very good statistical fit of the
empirical distribution; this is further verified by the
corresponding Q-Q plot.

\begin{figure}
 \begin{center}
 \includegraphics[width=6.cm,keepaspectratio=true]{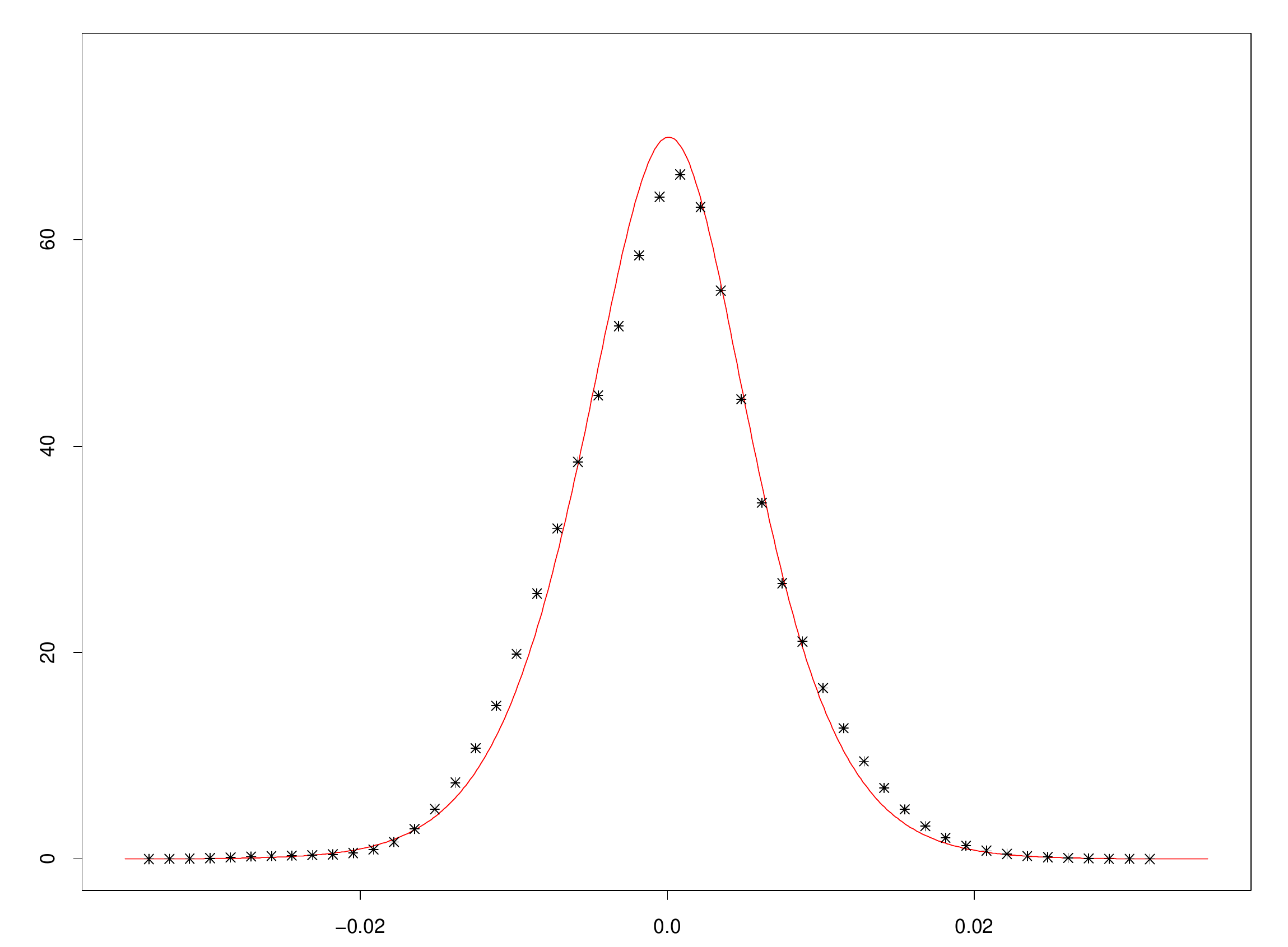}
 \includegraphics[height=6.cm,angle=90,keepaspectratio=true]{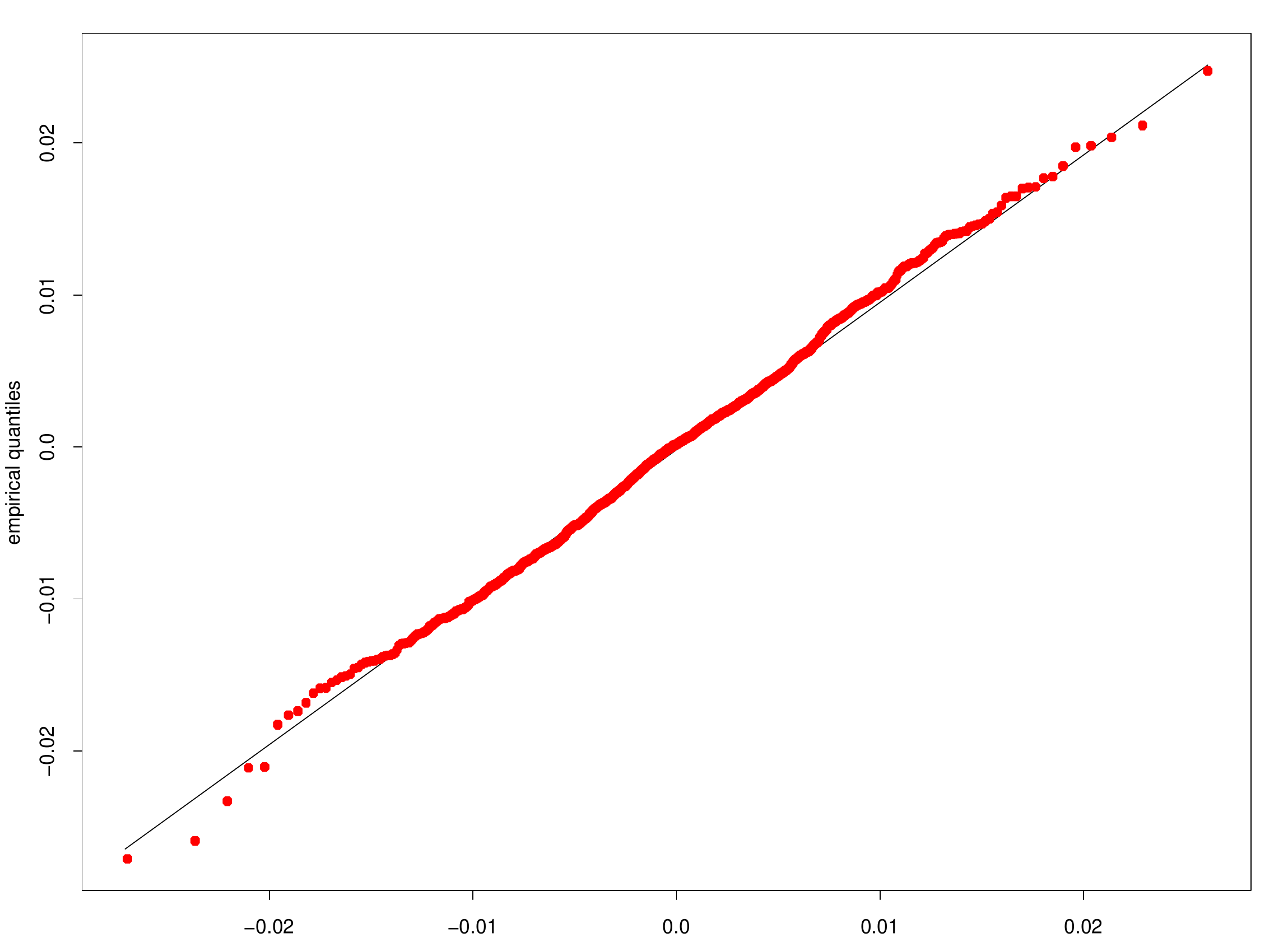}
   \caption{Empirical distribution and Q-Q plot of EUR/USD daily log-returns
           with fitted GH (red).}
   \label{GH-fit}
\end{center}
\end{figure}

Under the ``risk-neutral'' measure, the flexibility of the
generating distributions allows the implied volatility smiles
produced by a \lev model to accurately capture the shape of the
implied volatility smiles observed in the market. A typical
volatility surface can be seen in Figure \ref{surface}. Figure
\ref{smile-fit} exhibits the volatility smile of market data
(EUR/USD) and the calibrated implied volatility smile produced by
the NIG distribution; clearly, the resulting smile fits the data
particularly well.

\begin{figure}
\begin{center}
 \includegraphics[width=8.00cm,keepaspectratio=true]{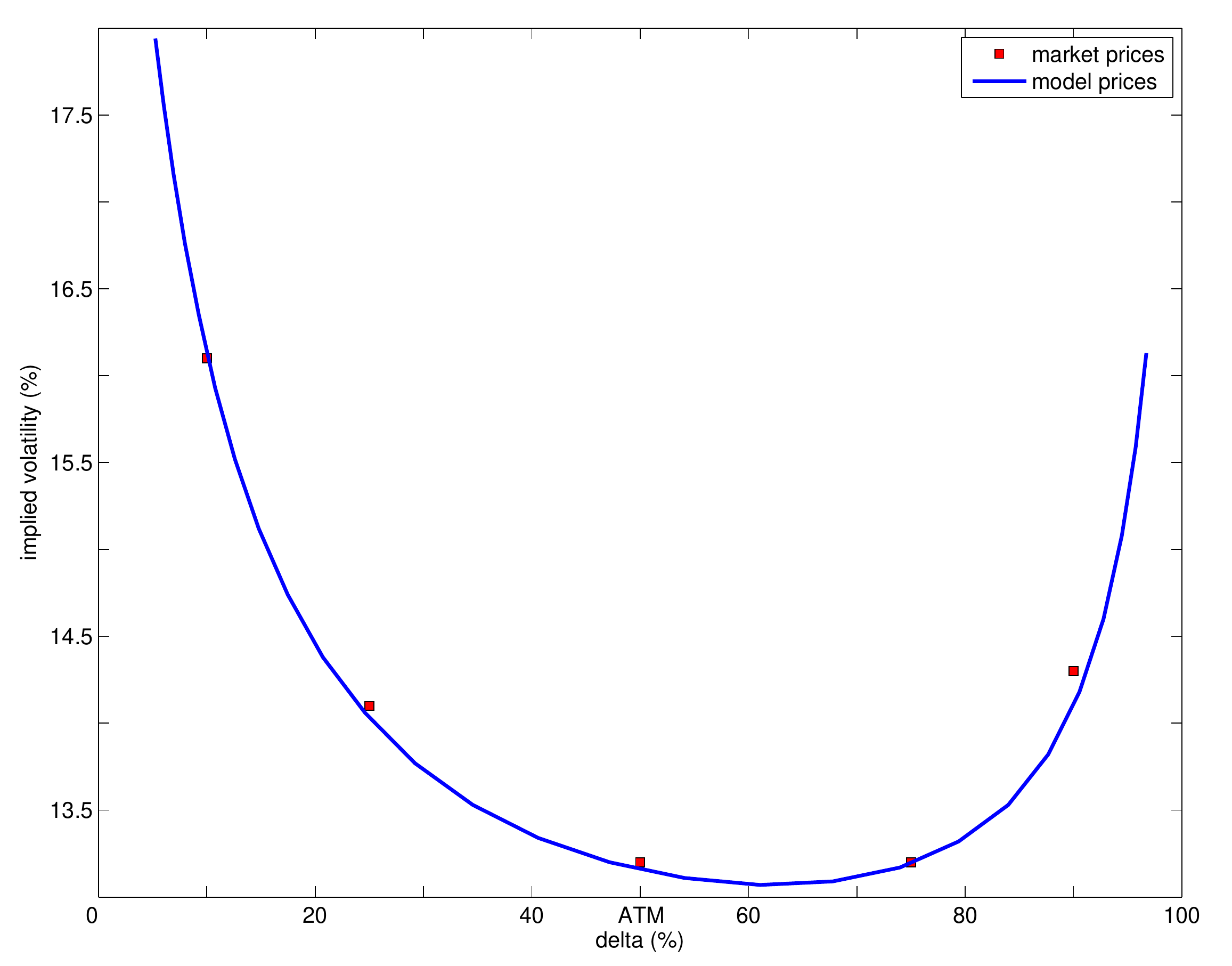}
  \caption{Implied volatilities of EUR/USD options and calibrated NIG smile.}
   \label{smile-fit}
\end{center}
\end{figure}

\appendix
\section{Poisson random variables and processes}

\begin{defin}
Let $X$ be a Poisson distributed random variable with parameter
$\lambda\in\R_{\geqslant0}$. Then, for $n\in\mathbb{N}$ the
probability distribution is
\begin{align*}
 P(X=n) = \e^{-\lambda}\frac{\lambda^n}{n!}
\end{align*}
and the first two centered moments are
\begin{align*}
 \text{E}[X]=\lambda
 &\qquad \text{and} \qquad
 \text{Var}[X]=\lambda.
\end{align*}
\end{defin}

\begin{defin}
A c\`{a}dl\`{a}g, adapted stochastic process \prozess[N] with
$N_t:\Omega\times\R_{\geqslant0}\rightarrow\mathbb{N}\cup\{0\}$ is
called a \emph{Poisson process} if
\begin{enumerate}
    \item $N_0=0$,
    \item $N_t - N_{s}$ is independent of $\F_s$ for any
          $0\leq s<t<T$,
    \item $N_t - N_{s}$ is Poisson distributed with parameter $\lambda(t-s)$
          for any $0\leq s<t<T$.
\end{enumerate}
Then, $\lambda\geq0$ is called the \textit{intensity} of the Poisson
process.
\end{defin}

\begin{defin}
Let $N$ be a Poisson process with parameter $\lambda$. We shall call
the process \prozess[\overline{N}] with
$\overline{N}_t:\Omega\times\R_{\geqslant0}\rightarrow\R$ where
\begin{align}\label{compensated-Poisson}
 \overline{N}_t:=N_t-\lambda t
\end{align}
a \emph{compensated Poisson} process.
\end{defin}

A simulated path of a Poisson and a compensated Poisson process can
be seen in Figure \ref{poisson-compoundpoisson}.
\begin{figure}
\begin{center}
 \includegraphics[width=6.cm,keepaspectratio=true]{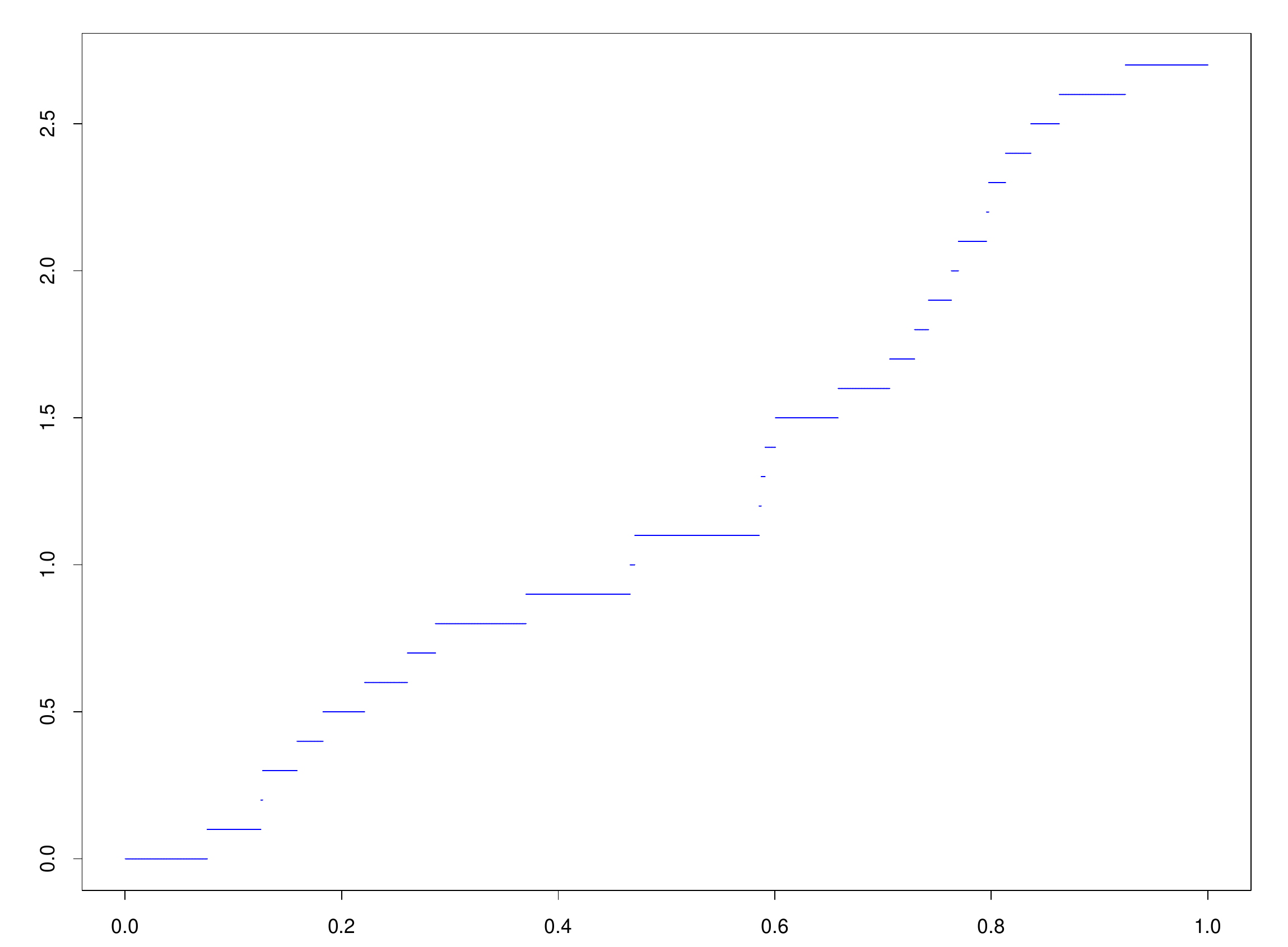}
 \includegraphics[width=6.cm,keepaspectratio=true]{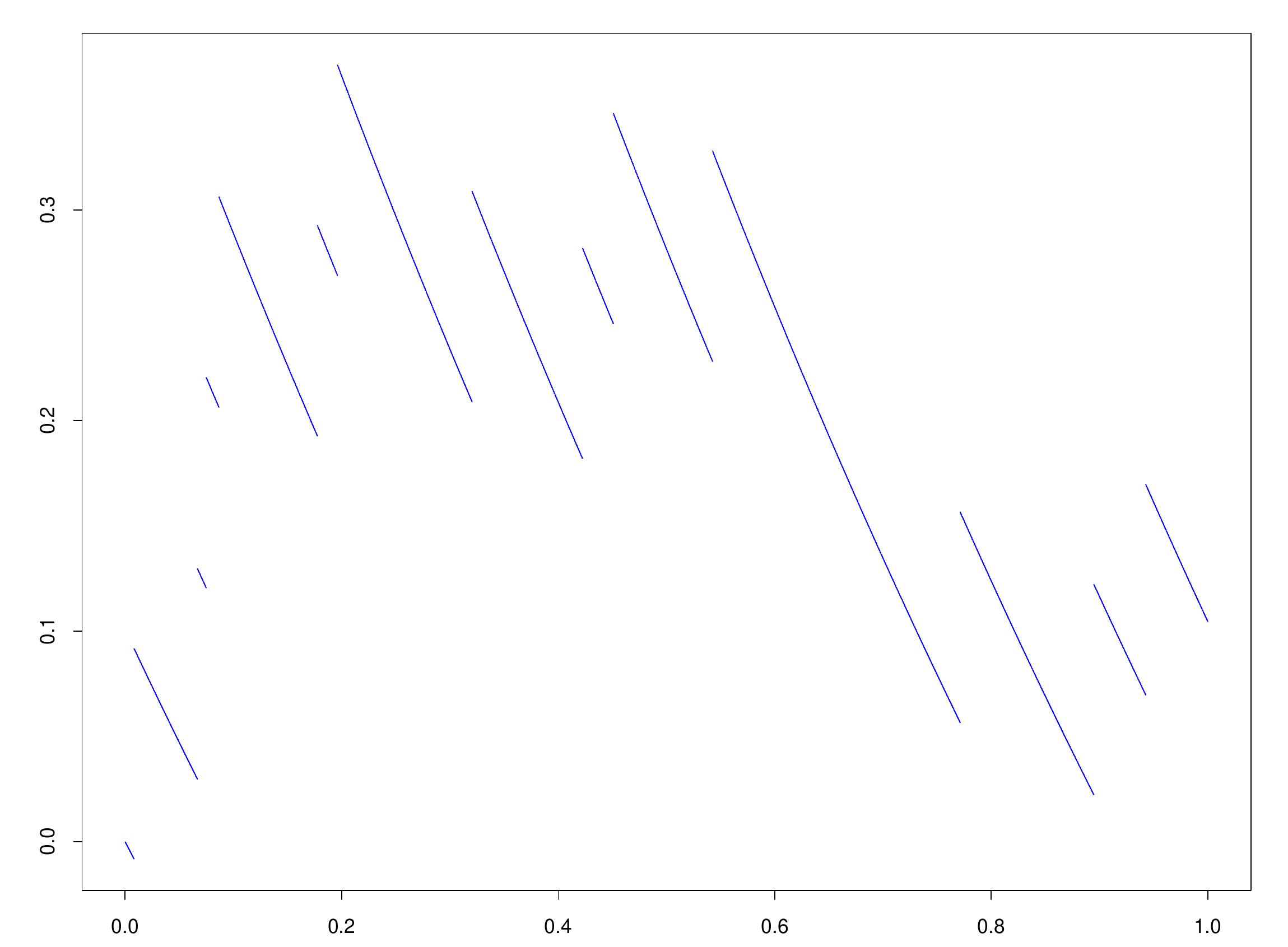}
  \caption{Plots of the Poisson (left) and compensated Poisson process.}
   \label{poisson-compoundpoisson}
\end{center}
\end{figure}

\begin{prop}
The compensated Poisson process defined by
\eqref{compensated-Poisson} is a martingale.
\end{prop}
\begin{proof}
We have that
\begin{enumerate}
\item the process $\overline{N}$ is adapted to the filtration
      because $N$ is adapted (by definition);
\item $\E[|\overline{N}_t|]<\infty$ because
      $\E[|N_t|]<\infty$, for all $\ott$;
\item finally, let $0\leq s<t<T$, then
  \begin{align*}
     \E[\overline{N}_t|\F_s] &= \E[N_t-\lambda t|\F_s]\\
        &= \E[N_s-(N_t-N_s)|\F_s] -\lambda (s-(t-s))\\
        &= N_s-\E[(N_t-N_s)|\F_s] -\lambda (s-(t-s))\\
        &= N_s-\lambda s\\
        &= \overline{N}_s. \qedhere
  \end{align*}
\end{enumerate}
\end{proof}

\begin{rem}
The characteristic functions of the Poisson and compensated Poisson
random variables are respectively
\begin{align*}
 \E[\e^{iuN_t}] = \exp\big[\lambda t(\e^{iu}-1)\big]
\end{align*}
and
\begin{align*}
 \E[\e^{iu\overline{N}_t}] = \exp\big[\lambda t(\e^{iu}-1-iu)\big].
\end{align*}
\end{rem}

\section{Compound Poisson random variables}\label{comp-Poisson}

Let $N$ be a Poisson distributed random variable with parameter
$\lambda\ge0$ and $J=(J_k)_{k\geq1}$ an i.i.d. sequence of random
variables with law $F$. Then, by conditioning on the number of jumps
and using independence, we have that the characteristic function of
a compound Poisson distributed random variable is
\begin{align*}
\E\Big[\e^{iu\sum_{k=1}^{N}J_k}\Big]
  &= \sum_{n\geq0}\E\Big[\e^{iu\sum_{k=1}^{N}J_k}\big|N=n\Big]P(N=n)\\
  &= \sum_{n\geq0}\E\Big[\e^{iu\sum_{k=1}^{n}J_k}\Big]\e^{-\lambda}\frac{\lambda^n}{n!}\\
  &= \sum_{n\geq0}\left(\int_{\R}\e^{iux}F(\ud x) \right)^n\e^{-\lambda}\frac{\lambda^n}{n!}\\
  &= \exp\left(\lambda\int_{\R}(\e^{iux}-1)F(\ud x)\right).
\end{align*}

\section{Notation}

$\eqlaw$ equality in law, $\mathcal L(X)$ law of the random variable
$X$\par $a\wedge b=\min\{a,b\}$, $a\vee b=\max\{a,b\}$\par
$\mathbb{C}\ni z=\alpha+i\beta$, with $\alpha,\beta\in\R$; then $\Re
z=\alpha$ and $\Im z=\beta$\par $1_A$ denotes the indicator of the
generic event $A$, i.e.
\begin{align*}
\1_A(x) = \left\{%
\begin{array}{ll}
    1, & \hbox{if $x\in A$,} \\
    0, & \hbox{if $x\notin A$.} \\
\end{array}%
\right.
\end{align*}
\indent\textit{Classes:}\par
$\mathcal M_\text{loc}$ local martingales\par
$\mathcal A_\text{loc}$ processes of locally bounded variation\par
$\mathcal V$ processes of finite variation\par
$G_\text{loc}(\mu)$ functions integrable wrt the compensated random measure $\mu-\nu$

\section{Datasets}

The EUR/USD implied volatility data are from 5 November 2001. The
spot price was 0.93, the domestic rate (USD) 5\% and the foreign
rate (EUR) 4\%. The data are available at
\begin{center}
  \texttt{http://www.mathfinance.de/FF/sampleinputdata.txt}.
\end{center}

The USD/JPY, EUR/USD and GBP/USD foreign exchange time series
correspond to noon buying rates (dates: $22/10/1997-22/10/2004$,
$4/1/99-3/3/2005$ and $1/5/2002-3/3/2005$ respectively). The data
can be downloaded from
\begin{center}
  \texttt{http://www.newyorkfed.org/markets/foreignex.html}.
\end{center}

\section{Paul L\'evy}

Processes with independent and stationary increments are named
\emph{\lev processes} after the French mathematician Paul L\'evy
(1886-1971), who made the connection with infinitely divisible laws,
characterized their distributions (L\'evy-Khintchine formula) and
described their path structure (L\'evy-It\^o decomposition). Paul
L\'evy is one of the founding fathers of the theory of
\textit{stochastic processes} and made major contributions to the
field of \textit{probability theory}. Among others, Paul L\'evy
contributed to the study of Gaussian variables and processes, the
law of large numbers, the central limit theorem, stable laws,
infinitely divisible laws and pioneered the study of processes with
independent and stationary increments.

More information about Paul \lev and his scientific work, can be
found at the websites
\begin{center}
  \texttt{http://www.cmap.polytechnique.fr/~rama/levy.html}
\end{center}
and
\begin{center}
  \texttt{http://www.annales.org/archives/x/paullevy.html}
\end{center}
(in French).

\section*{Acknowledgments}

A large part of these notes was written while I was a Ph.D. student
at the University of Freiburg; I am grateful to Ernst Eberlein for
various interesting and illuminating discussions, and for the
opportunity to present this material at several occasions. I am
grateful to the various readers for their comments, corrections and
suggestions. Financial support from the Deutsche
Forschungsgemeinschaft (DFG, Eb\,66/9-2) and the Austrian Science
Fund (FWF grant Y328, START Prize) is gratefully acknowledged.

\bibliographystyle{chicago}
\bibliography{/home/famuser/papapan/Papers/references}

\end{document}